\documentclass[a4paper,12pt,justification=centering]{article}
\usepackage{latexsym}
\usepackage{rotating}
\usepackage{amsmath,amssymb,amsxtra,amscd,amsthm,amsfonts, stmaryrd}
\usepackage{tikz}
\usetikzlibrary{calc,decorations.markings}
\usepackage[mathscr]{eucal}
\usepackage{graphicx}
\usepackage[justification=centering]{caption}
\usepackage{subfig}
\usepackage{datetime}
\usepackage{pdfsync}
\usepackage{thumbpdf}
\usepackage{color}
\ifx\pdfoutput\undefined

\usepackage{amsmath}
\usepackage{amssymb}
\usepackage{pgfplots}
\usepackage{graphicx} 
\usepackage{caption}
\usepackage{ mathrsfs }
\usepackage{subcaption}
\usepackage{color}
\usepackage[nottoc,notlot,notlof]{tocbibind}

\usepackage{verbatim}
\usepackage{tensor}
\usepackage{xcolor}
\usepackage{titlesec}

\graphicspath {{images/}}
\usepackage{wrapfig}

\usepackage
[dvips,
 verbose=false,
 bookmarks=true,
 colorlinks=true,
 linkcolor=webred,
 filecolor=webbrown,
 citecolor=webgreen,
 pagecolor=webblue,
 urlcolor=webblue,
 pdftitle={},
 pdfauthor={Murad Alim},
 pdfsubject={},
 pdfkeywords={},
 bookmarksopen=false,
 pdfpagemode=None,
 pdfview=FitH,
 pdfstartview=FitH,
 extension=pdf]{hyperref}
\else
\usepackage
[pdftex,
 verbose=false,
 bookmarks=true,
 colorlinks=true,
 linkcolor=webred,
 filecolor=webbrown,
 citecolor=webgreen,
 pagecolor=webblue,
 urlcolor=webblue,
 pdftitle={},
 pdfauthor={Murad Alim},
 pdfsubject={},
 pdfkeywords={},
 bookmarksopen=false,
 pdfpagemode=None,
 pdfview=FitH,
 pdfstartview=FitH,
 extension=pdf]{hyperref}
\fi
\definecolor{webred}{rgb}{.8,0,0}
\definecolor{webbrown}{rgb}{.6,0,0}
\definecolor{webgreen}{rgb}{0,0.5,0}
\definecolor{webdkgreen}{rgb}{0,0.3,0}
\definecolor{webblue}{rgb}{0,0,0.5}

\numberwithin{equation}{section}

\providecommand{\href}[2]{#2}

\allowdisplaybreaks
\usepackage{bbm}
\usepackage{rotating}
\newcommand{\be}{\begin{eqnarray}}
\newcommand{\beq}{\begin{eqnarray}}
\newcommand{\ee}{\end{eqnarray}}

\newcommand{\BR}{{\mathbb R}}

\newcommand{\ii}{\textnormal{i}}

\newcommand{\I}{{\mathrm i}}
\newcommand{\la}{{\lambda}}

\def\e#1\e{\begin{equation}#1\end{equation}}
\def\ea#1\ea{\begin{align}#1\end{align}}

\theoremstyle{plain}
\newtheorem{thm}{Theorem}[section]
\newtheorem*{thm*}{Theorem}

\newtheorem{prop}[thm]{Proposition}

\theoremstyle{definition}
\newtheorem{dfn}[thm]{Definition}

\newtheorem{rem}[thm]{Remark}

\begin{document}

\setlength{\parindent}{0cm}
\setlength{\baselineskip}{1.5em}
\title{\bf{Non-perturbative topological strings\\ from resurgence}}

\author{Murad Alim\footnote{\tt{murad.alim@tum.de}}\\
\small Department of Mathematics, Technical University of Munich, Boltzmannstr. 3, 85748 Garching\\
\small Maxwell Institute for Mathematical Sciences, Edinburgh EH14 4AS, UK\\
\small Department of Mathematics, Heriot-Watt University, Edinburgh EH14 4AP, UK}

\date{}
\maketitle

\abstract{The partition function of topological string theory on any family of Calabi-Yau threefolds is defined perturbatively as an asymptotic series in the topological string coupling and encodes, in a holomorphic limit, higher genus Gromov-Witten as well as Gopakumar-Vafa invariants. We prove that the partition function of topological strings of any CY in this limit can be written as a product, where each factor is given by the partition function of the resolved conifold with shifted arguments, raised to the power of certain sheaf invariants. We use this result to put forward an expression for the non-perturbative topological string partition function in this limit, as a product over analytic functions in the topological string coupling which correspond to the Borel sums for the resolved conifold found previously. We furthermore find an expression for the Borel transform of the full asymptotic series in this limit expressed in terms of the sheaf invariants. We use this to define the Borel sums and compute the corresponding Stokes jumps which constitute non-perturbative corrections to the partition function. The jumps depend only on genus zero GV invariants and their sum can be expressed entirely in terms of a single function which is introduced as a deformation of the prepotential.}

\clearpage


\tableofcontents

\section{Introduction and summary}
\subsection{Introduction}
Topological string theory is defined perturbatively and its partition function is given by an asymptotic series in the topological string coupling. The perturbative definition as well as the asymptotic nature of the expansion of its partition function are common traits of quantum field and string theories, see e.~g.~\cite{Marino:2012zq} and references therein. The clear mathematical formulation of topological strings as well as its connections to enumerative geometry make however both perturbative computations as well as the study of its non-perturbative structure highly accessible.

Perturbative computations within topological strings are accessed by the holomorphic anomaly equations of \cite{Bershadsky:1993cx} and their recursive solution in terms of Feynman diagrams. The feasibility of higher genus computation was significantly enhanced due to the discovery of a polynomial structure of the free energies \cite{Yamaguchi:2004bt,Alim:2007qj} \footnote{See also \cite{Alimlectures} and references therein } which can be used together with expected boundary conditions to access higher genus computations \cite{Huang:2006hq}, see also \cite{Alexandrov:2023zjb} for a recent study. The anomaly equations were also used to access the non-perturbative structure by making a suitable ansatz using transseries in \cite{Couso-Santamaria:2013kmu,Couso-Santamaria:2014iia}, see also the recent \cite{Gu:2023mgf}. The non-perturbative structure of topological strings has been addressed in many works, using connections and dualities of topological string theory with other areas of physics and mathematics such as the study of black holes partition functions \cite{Ooguri:2004zv}, large order behavior \cite{Marino:2006hs,Marino:2007te}, matrix models \cite{Eynard:2008he,Pasquetti:2009jg}, superconformal theories \cite{Lockhart:2012vp}, quantum mechanical spectral theory \cite{Grassi:2014zfa} and Chern-Simons theories \cite{Krefl:2015vna} to name a few, see \cite{Aniceto:2011nu,Marino:2012zq} and references therein.

The focus of this work is on the non-perturbative structure of topological string theory as addressed via the theory of resurgence, which was developed by \'Ecalle \cite{ecalle1981resurgent},  see \cite{SauzinLectures} for an introduction to the subject as well as \cite{Marino:2012zq,Marino:2024tbx} for applications in physics and topological string theory. The use of resurgence for the study of non-perturbative topological string theory was initiated in \cite{Marino:2008ya}. Borel summation and resurgence of topological strings for geometries which can be related to matrix models including the resolved conifold was studied in \cite{Pasquetti:2009jg}, the Borel resummation for the resolved conifold was further studied in \cite{Hatsuda:2015owa}. The study of Borel resummation for the resolved conifold was revisited in \cite{Alim:2021mhp} building on methods and results of \cite{Coman:2018uwk,Coman:2020qgf,Alim:2020tpw,Alim:2021ukq,Garoufalidis:2020pax}  and obtaining exact analytical results for the Borel transform and Stokes jumps and identifying the Stokes jump coefficients with BPS or Dondaldson-Thomas (DT) invariants using the link to the Riemann-Hilbert problem studied in \cite{Bridgeland1,BridgelandCon}. In the series of works \cite{Gu:2022sqc,Iwaki:2023cek,Gu:2023mgf}, the resurgence of topological string theory on arbitrary Calabi-Yau threefolds has been studied via the holomorphic anomaly equations satisfied by the free energies, producing formal solutions for their multi-instanton trans-series extensions. The Stokes constants in these resurgent structures are conjecturally identified with enumerative invariants counting BPS states of the topological string. The resurgence of open and closed as well as refined topological strings on this geometry was further studied in \cite{Grassi:2022zuk,Alim:2022oll}.  Further recent studies addressing resurgence of open and closed topological strings include \cite{Rella:2022bwn,Marino:2023nem,Alexandrov:2023wdj,Fantini:2024snx}

In this work we start with the general form of the topological string partition function encoding the Gromov-Witten and Gopakumar-Vafa invariants, which corresponds to a limit of the full topological string partition function which satisfies the holomorphic anomaly equation. By using sheaf invariants put forward by Maulik and Toda \cite{MaulikToda} to define mathematically the GV invariants, we bring the partition function into a form which can be recognised as an infinite product of a fundamental building block for every sheaf invariant which is evaluated at a shifted argument and raised to the power of the corresponding sheaf invariant. The building block is the topological string partition function of the resolved conifold. We then use the analytic functions in the topological string coupling which were obtained as Borel sums in the resurgence analysis of \cite{Alim:2021mhp} of the resolved conifold to put forward a non-perturbative expression for the topological string partition function in the holomorphic limit in the general case. We then proceed by studying the resurgence of the asymptotic series of topological strings in this limit expressed in terms of the sheaf invariants. We prove an expression for the Borel transform in this case and use this to define the Borel sums and compute the Stokes jumps. The sum of the latter constitute non-perturbative corrections to the topological string free energy and can be expressed in terms of a function which we introduce as a deformation of the prepotential and which plays the analogous role as the topological string partition function in the Nekrasov-Shatashvili limit \cite{Nekrasov:2009rc} in the case of the resolved conifold.

\subsection{Summary}
Before proceeding we give a brief summary of the results. The main results of this paper are an expression for the non-perturbative topological string partition function on any Calabi-Yau manifold $X$ as well as an expression for the Borel transform, Borel sums and Stokes jumps. These expressions hold in the limit of topological string theory which contains the enumerative invariants of the underlying CY manifold.

The analytic expression is given, up to the contributions from constant maps at genus $0,1$, by:
\begin{equation}
Z_{\textrm{top},np}(\lambda,t) =  Z_{c,np}(\lambda)  \prod_{n\in\mathbb{Z}} \prod_{\beta>0} \left( Z^{con}_{np} (\lambda, t^{\beta}+ n \, \check{\lambda})\right)^{\Omega_{\beta,n}}\,,
\end{equation}
where
$$Z^{con}_{np}(\lambda,t)= \prod_{m=0}^{\infty} (1-Q\, q^{m+1})^{m+1}  \cdot \exp\left(-\frac{1}{2\pi i} \operatorname{Li}_2(Q' q'^m)\right) \cdot (1-Q' q'^m)^{-\frac{2\pi}{\lambda}\left( t+m\right)}\,, $$
$$
Z_{c,np}(\lambda)= \left( \prod_{m=1}^{\infty} (1- q^{m})^{m}  \cdot \exp\left(-\frac{1}{2\pi i} \operatorname{Li}_2(q'^m)\right) \cdot (1-q'^m)^{-\frac{ m}{\check{\lambda}}}\right)^{-\chi(X)/2}\,,
$$
where $Q=\exp(2\pi i t),\quad q=e^{\ii\lambda}\,, \quad Q'=\exp(2\pi i t/\check{\lambda})\,, \quad q'=\exp(2\pi i/\check{\lambda})\,,\quad \check{\lambda}=\frac{\lambda}{2\pi}$ and where the subscript $c$ in $Z_c$ stands for the constant map contribution. The class $\beta \in H_{2}(X,\mathbb{Z})$ refers to a curve class, and $\beta>0$ denotes the non-constant curves. $\Omega_{\beta,n}$ are up to a sign the sheaf invariants defined in terms of the higher genus Gopakumar-Vafa invariants as follows \cite{MaulikToda} \footnote{We follow the notation of \cite{Stoppa}}
\begin{equation}
\sum_{n\in \mathbb{Z}\,} \Omega_{\beta,n} y^n = \sum_{g\ge0} [\textrm{GV}]_{\beta,g}\, (-1)^g (y^{1/2}-y^{-1/2})^{2g}\,.
\end{equation}

The expression can be given in terms of special functions as follows:

\begin{equation}
\begin{split}
Z_{\textrm{top},np}(\lambda,t) =\left( \mathcal{S}_3(0 | \check{\lambda},1)\right)^{-\chi(X)/2}  \prod_{n\in\mathbb{Z}} \prod_{\beta>0} \left( \mathcal{S}_3(t^{\beta} + n \check{\lambda} | \check{\lambda},1)\right)^{\Omega_{\beta,n}}\,,
\end{split}
\end{equation}
where 
$$ 
    \mathcal{S}_3(z\, | \, \omega_1,\omega_2) := \exp\left(\frac{\pi \I}{6} \cdot B_{3,3}(z+\omega_1\,|\,\omega_1,\omega_1,\omega_2)\right) \cdot \sin_3(z+\omega_1\, |\, \omega_1,\omega_2,\omega_3),$$
and the relevant definitions for the Bernoulli polynomials and the triple sine function are given in the appendix \ref{appendix}.

This proposal is put forward by using theorem  \ref{thm:conifolds} of this paper, stating that  the generating function of GW invariants, excluding the classical terms as well as the contributions from constant maps for any Calabi-Yau threefold:
\begin{equation}
\tilde{F}_{GW}(\lambda,t)= \sum_{g=0}^{\infty}\sum_{\beta>0}  \lambda^{2g-2} [GW]_{\beta,g} Q^{\beta}\,,
\end{equation} 
can be written as:
\begin{align}
\tilde{F}_{GW}(\lambda,t)&=  \sum_{n\in \mathbb{Z}}  \sum_{\beta>0}   \Omega_{\beta,n} \widetilde{F}^{con} (\lambda, t^{\beta,n}) \, \\
&=\sum_{n\in \mathbb{Z}}  \sum_{\beta>0}   \Omega_{\beta,n}  \left( \frac{1}{\lambda^2} \mathrm{Li}_{3}(Q^{\beta,n})+ \sum_{g=1}^{\infty} \lambda^{2g-2} \frac{(-1)^{g-1}B_{2g}}{2g (2g-2)!} \mathrm{Li}_{3-2g} (Q^{\beta,n}) \, \right)\,,
\end{align}
where $t^{\beta,n}:=t^{\beta}+ n \check{\lambda}\,,Q^{\beta,n}:=\exp(2\pi \ii t^{\beta,n})$  and then using the Borel summation analysis of \cite{Alim:2021mhp}.

Furthermore the non-perturbative partition function is written as:
\begin{align}
Z_{\textrm{top},np}(\lambda,t) =  Z_{c,np}(\lambda) \times Z_{top,wk}(\lambda,t) \times Z_{top,sg}(\lambda,t)
\end{align}
where
\begin{align}
 Z_{top,wk}(\lambda,t)&:= \prod_{n\in \mathbb{Z}}  \prod_{\beta>0}  \prod_{m=1}^{\infty}   \left(1-   Q^{\beta}   \exp(\ii (n+m) \,\lambda) \right)^{m\, \Omega_{\beta,n}}\,,\\
 Z_{top,sg}(\lambda,t) &:=   \prod_{\beta>0} \exp\left(-\frac{[\textrm{GV}]_{\beta,0}}{2\pi i} \operatorname{Li}_2(Q'^{\beta} q'^m)\right) \cdot (1-Q'^{\beta} q'^m)^{-\frac{ [\textrm{GV}]_{\beta,0}}{\check{\lambda}}\left( t^{\beta}+m\right)}\,,
\end{align}
where the subscripts $wk$ and $sg$ stand for weak ($\lambda$ small) and strong ($1/\lambda $ small) respectively and denote the factors of the non-perturbative partition function which are expressions in $q=\exp(\ii \lambda)$ and $q'=\exp(2\pi \ii /\check{\lambda})$ respectively. The surprising result is that the non-perturbative corrections only depend on genus $0$ GV invariants which is due do cancellations leading to a dependence on only the sum of the sheaf invariants for each curve class $\beta$, see \eqref{eq:sumOmega}:
\begin{equation}
\sum_{n\in \mathbb{Z}} \Omega_{\beta,n}  = [\textrm{GV}]_{\beta,0}\,,
\end{equation}
where the sum on the LHS has only finitely many non-zero terms due to the finiteness of GV invariants, see sec.\ref{sec:finiteness}.

 The non-perturbative corrections can be furthermore packaged in terms of an $\varepsilon-$deformation of the prepotential, which we define as follows:
\begin{equation}
\tilde{F}^{0}_{\textrm{def}}(\varepsilon,t):=\frac{1}{2} \sum_{k,\beta>0} \frac{[\textrm{GV}]_{\beta,0}}{k^2} \frac{Q^{k\beta}}{\sin(k\varepsilon/2)}\,.
\end{equation}
 The $q'$ dependent factor
 $$Z_{top,sg}(\lambda,t)= \exp\left( F_{top,sg}(\lambda,t)\right)\,,$$
 can be expressed as
 \begin{equation}
F_{top,sg}(\lambda,t)=  -\frac{1}{2\pi} \partial_{\lambda} \left( \lambda \, \tilde{F}^{0}_{\textrm{def}} \left( \frac{2\pi}{\check{\lambda}}, \frac{t-1/2}{\check{\lambda}} \right)\right)\,.
\end{equation}
This $\varepsilon$ deformation of the prepotential which we propose agrees with the refined topological string in the Nekrasov-Shatashvili limit in the case of the resolved conifold.

We then proceed to study the non-perturbative structure of topological string theory in this limit via resurgence, we therefore express the asymptotic series of higher genus GW invariants in terms of the sheaf invariants finding in particular new expressions of the higher genus GW generating functions:

\begin{equation}
\tilde{F}^g(t)= \sum_{\beta>0} \sum_{n\in\mathbb{Z}}\frac{(-1)^g \Omega_{\beta,n}}{2\cdot2g!} \left( (2g-1) (B_{2g}(n)+B_{2g}(-n))- 2g n (B_{2g-1}(n)-B_{2g-1}(-n))\right)  \operatorname{Li}_{3-2g}(Q^{\beta}) \,,
\end{equation}

for $g\ge 2$, where $B_m(x)$ denotes Bernoulli polynomials. For $\Phi(\check{\lambda},t):= \sum_{g=2}^{\infty} \lambda^{2g-2} \tilde{F}^g(t), \check{\lambda}=\lambda/2\pi$  we show in theorem \ref{thm:Borel} that the Borel transform can analytically continued to a meromorphic function given by:
\begin{equation} 
G(\xi,t) 
=  \frac{1}{\xi} \sum_{\beta>0} \sum_{k,n\in \mathbb{Z}} \frac{\Omega_{\beta, n}}{2}
\left(
  \frac{2(t^{\beta}+k)^2}{\xi^2}
  + \left(n^2 - \frac{1}{6}\right)
  - \frac{
      e^{\frac{\xi(n+1)}{t^{\beta}+k}}
      + e^{\frac{\xi(-n+1)}{t^{\beta}+k}}
    }{
       \left( e^{\frac{\xi}{t^{\beta}+k}} - 1 \right)^2
    } 
\right)\,,
\end{equation}
with poles at $\xi= 2\pi \ii m  (t^\beta+k)\,,$ for $m\in \mathbb{Z}\setminus \{0\}$ and $k\in \mathbb{Z}$. Theorem \ref{thm:Stokes} then provides expressions for the Stokes jumps and identifies an infinite sum of these with the non-perturbative corrections obtained from the use of the resolved conifold results.

The organization of this work is as follows. In sec.~\ref{sec:freeenergies}, the topological string free energies are recalled as well as their Gromov-Witten and Gopakumar-Vafa expansions. We proceed with a review of the resurgence analysis results of \cite{Alim:2021mhp} to introduce the building blocks of the more general topological string analysis in later sections, we provide new infinite product expressions for the non-perturbative partition function of topological strings on the resolved conifold as well as for topological strings on $\mathbb{C}^3$ which is also the building block for the more general constant maps contributions. The latter also provides an interesting non-perturbative completion of the MacMahon function.


\section{Topological strings and enumerative geometry}\label{sec:freeenergies}
To a mirror family of CY threefolds, topological string theory associates the topological string partition function which is defined as an asymptotic series in the topological string coupling $\lambda$, summing over the free energies $ \mathcal{F}^{g}(t)$ associated to Riemann surfaces of genus $g$:
\begin{equation}
Z_{top} (\lambda,t)= \exp \left(\sum_{g=0}^{\infty} \lambda^{2g-2} \mathcal{F}^{g}(t)\right)\,,
\end{equation}
where $t=(t^1,\dots,t^n)$ is a set of distinguished local coordinates on the underlying moduli space $\mathcal{M}$, which is of dim $n=h^{1,1}(X_t)=h^{2,1}(\check{X}_{t(z}))$. $X_t$ and $\check{X}_{t(z)}$ are a mirror pair of CY threefolds which correspond to the A-model and B-model sides of mirror symmetry, the notation $t(z)$ refers to the mirror map which is used to identify a natural set of coordinates on both sides.


\subsection{The Gromov-Witten potential}
In a certain limit together with an expansion around a distinguished large volume point in the moduli space, see e.~g.~\cite{Bershadsky:1993cx} the topological string free energies $\mathcal{F}^{g}(t)$, which are non-holomorphic sections of certain line bundles, become the generating functions of higher genus Gromov-Witten invariants $F^g(t)$ on the A-model side of mirror symmetry. The GW potential of $X$ is the following formal power series:
\begin{equation}
F_{GW}(\lambda,t) = \sum_{g\ge 0}  \lambda^{2g-2} F^g(t)= \sum_{g\ge 0}  \lambda^{2g-2} \sum_{\beta\in H_2(X,\mathbb{Z})}  [\textrm{GW}]_{\beta,g} \,Q^{\beta}\, ,
\end{equation}
where $Q^{\beta} := \exp (2\pi \I t^{\beta} )$ is a formal variable and $[\textrm{GW}]_{\beta,g}$ are the Gromov-Witten invariants associated to a curve of genus $g$ and class $\beta$. Let $C_a \in H_{2}(X,\mathbb{Z}), a=1,\dots,n=h^{1,1}(X)$ be a set of curve classes spanning $H_2(X,\mathbb{Z})$, then in general all curve classes $\beta$ can be written as $\beta=\sum_{a=1}^{n} d_a C_a$, where the $d$'s are the degrees of the curve. The formal variables $t^{\beta}$ become $t^{\beta}=\sum_{a=1}^{n} d_a t^a$ and the $t^a$ correspond to local coordinates on the moduli space of complexified K\"ahler forms of the underlying CY $X$. Bearing this in mind, we will in the following treat $t^{\beta}$ as independent formal variables for the curve classes $\beta$ for ease of exposition of the general results.

The GW potential can be furthermore written as:
\begin{equation}
F_{GW}=F_c + \tilde{F}_{GW}\,,
\end{equation}
where $F_c$ denotes the contribution from constant maps and $ \tilde{F}$ the contribution from non-constant maps. The constant map contribution at genus 0 and 1 are $t$ dependent and the higher genus constant map contributions take the universal form \cite{Faber}:
\begin{equation}
F_c^g:= \frac{\chi(X)(-1)^{g-1}\, B_{2g}\, B_{2g-2}}{4g (2g-2)\, (2g-2)!}\,, \quad g\ge2\,,
\end{equation}
where $\chi(X)$ is the Euler characteristic of $X$ and the Bernoulli numbers $B_n$ are generated by:
\begin{equation}
\frac{w}{e^w-1} = \sum_{n=0}^{\infty} B_n \frac{w^n}{n!}\,.
\end{equation}
\subsection{The Gopakumar-Vafa resummation}

The Gopakumar-Vafa (GV) resummation of the GW potential \cite{Gopakumar:1998ii,Gopakumar:1998jq} reformulates the non-constant part of the GW potential in terms of the Gopakumar-Vafa invariants  $[\textrm{GV}]_{\beta,g} \in \mathbb{Z}$ which are given by a count of electrically charged $M_2$ branes in an M-theory setup. The GW potential can thus be written as:
\begin{equation}\label{GVresum}
F_{GV}(\lambda,t)= \sum_{\beta>0}\sum_{g\ge 0} [\textrm{GV}]_{\beta,g}\, \sum_{k\ge 1} \frac{1}{k} \left( 2 \sin \left( \frac{k\lambda}{2}\right)\right)^{2g-2} Q^{k\beta}\, .
\end{equation}
In particular, the asymptotic expansion of $F_{GV}$ at  $\lambda=0$ is given by.$\tilde{F}_{GW}$. One can in particular re-write the genus $0$ part of $\tilde{F}_{GW}$ in terms of the GV invariants as 
$$ \tilde{F}^0(t)=\sum_{\beta>0} [\textrm{GV}]_{\beta,0}\, \operatorname{Li}_3(Q^{\beta})\,, \quad Q^{\beta}= \exp(2\pi i t^{\beta})\,.$$ 


\subsection{Finiteness of GV invariants}\label{sec:finiteness}
We remark that in the GV resummation of the GW potential the integrality of the $[\textrm{GV}]_{\beta,g}$ was conjectured from their M-theoretic BPS counting origin. The integrality was proven in \cite{Ionel}. A further conjecture concerns the finiteness of the GV invariants, meaning that for every curve class $\beta$, there exists a $g_{\beta}\in \mathbb{N}$ such that 

$$[\textrm{GV}]_{\beta,g'}=0 \,, \quad \forall g'\ge g_{\beta}\,,$$
see \cite{Doan} for a proof of this using methods from symplectic geometry and only applies to the class of CY geometries studied in that work. Algebraic geometrically it is expected that the $[\textrm{GV}]_{\beta,g}$ vanish for each curve class $\beta$ at high enough genus, this is discussed as the Castelnuovo bound for instance in \cite{Katz:1999xq} and is used as a boundary condition together with the polynomiality of \cite{Yamaguchi:2004bt}, for higher genus GV computations in the quintic \cite{Huang:2006hq}.


\subsection{GW potential of the resolved conifold}
An example of a non-compact CY manifold $X$  for which the asymptotic series is known to all orders is the 
resolved conifold which is given by the total space of the rank two bundle over the projective line:
\begin{equation}
X := \mathcal{O}(-1) \oplus \mathcal{O}(-1) \rightarrow \mathbb{P}^1\,,
\end{equation}
and corresponds to the resolution of the conifold singularity.\\ 

The GW potential for this geometry was determined in physics \cite{Gopakumar:1998ii,GV}, and in mathematics \cite{Faber} with the following outcome for the non-constant maps:
\begin{equation}\label{resconfree}
\widetilde{F}^{con}(\lambda,t)= \sum_{g=0}^{\infty} \lambda^{2g-2} \widetilde{F}^g(t)= \frac{1}{\lambda^2} \mathrm{Li}_{3}(Q)+ \sum_{g=1}^{\infty} \lambda^{2g-2} \frac{(-1)^{g-1}B_{2g}}{2g (2g-2)!} \mathrm{Li}_{3-2g} (Q) \, ,
\end{equation}
using the notation $Q=e^{2\pi\ii t}$.


\section{Building blocks of non-perturbative topological strings} \label{sec:buildingblocks}

In the following we review the resurgence study of the building blocks of topological strings which will be used in the next sections to spell out the generalization for topological strings on arbitrary CY manifolds. The building blocks in this section will be topological strings on  the resolved conifold as well as on $\mathbb{C}^3$ which also gives the piece of the topological string partition function stemming from constant maps.  The contents of this section build on \cite{Alim:2021mhp} with a few new important additions.

\subsection{The resolved conifold}\label{subsec:rescon}

The starting point is the formal series giving the non-constant and non-classical piece of the Gromov-Witten potential or the topological string free energy of the resolved conifold:
\begin{align*}\label{formal}
\widetilde{F}^{con}(\lambda,t)&=  \frac{1}{\lambda^2} \mathrm{Li}_{3}(Q)+\frac{B_2}{2}\mathrm{Li}_1(Q) + \sum_{g=2}^{\infty} \lambda^{2g-2} \frac{(-1)^{g-1}B_{2g}}{2g (2g-2)!}\, \mathrm{Li}_{3-2g} (Q) \,  \\
&=\frac{1}{\lambda^2} \mathrm{Li}_{3}(Q)  +\frac{B_2}{2}\mathrm{Li}_1(Q)+ \Phi^{con}(\check{\lambda},t)\,, \quad \check{\lambda}=\frac{\lambda}{2\pi}\,, \quad Q=e^{2\pi \I t}\,. 
\end{align*} 

In \cite{Alim:2020tpw} following methods of \cite{Iwaki2}, it was proven that this series satisfies the following difference equation:
\begin{equation}\label{eq:diffeq}
\widetilde{F}^{con}\left(\lambda,t+\check{\lambda}\right) + \widetilde{F}^{con}\left(\lambda,t-\check{\lambda}\right) - 2 \widetilde{F}^{con}\left(\lambda,t\right)=-\operatorname{Li}_1(Q) \,, \quad \check{\lambda}=\frac{\lambda}{2\pi}\,.\\
\end{equation}

A solution in terms of the triple sine function of this difference equation was found \cite{Alim:2021lld}:
\begin{equation}
\begin{split}
\widetilde{F}^{con}_{\text{np}}(\lambda,t) := \log \left( \mathcal{S}_3(t | \check{\lambda},1)\right)\,,
\end{split}
\end{equation}
where 
$$ 
    \mathcal{S}_3(z\, | \, \omega_1,\omega_2) := \exp\left(\frac{\pi \I}{6} \cdot B_{3,3}(z+\omega_1\,|\,\omega_1,\omega_1,\omega_2)\right) \cdot \sin_3(z+\omega_1\, |\, \omega_1,\omega_2,\omega_3),$$
and the relevant definitions for the Bernoulli polynomials and the triple sine function are given in the appendix \ref{appendix}. The non-perturbative content of this solution was analyzed in \cite{Alim:2021ukq} and in \cite{Alim:2021mhp}. In \cite{Alim:2022oll} it was furthermore shown that $\widetilde{F}^{con}_{\text{np}}(\lambda,t)$ satisfies a further difference equation given by:
\begin{align}\label{eq:diffeq2}
    \widetilde{F}_{\text{np}}^{con}(\lambda,t+1)-\widetilde{F}_{\text{np}}^{con}(\lambda,t)=\frac{1}{2\pi \I} \frac{\partial}{\partial \check{\lambda}} \left(  \check{\lambda} \, \operatorname{Li}_2(e^{2\pi \I t/\check{\lambda}})\right)\,.
\end{align}

$\widetilde{F}^{con}_{\text{np}}(\lambda,t)$  was identified in \cite{Alim:2021mhp} as the Borel summation of the asymptotic series along a distinguished ray on the real axis in the Borel plane, furthermore the inhomogeneous piece of the difference equation  \eqref{eq:diffeq2} corresponds to the Stokes jumps of the Borel summation of the asymptotic series. To clarify this we recall here the general story of Borel summation.

Given a formal power series $a(\check{\lambda}) \in \check{\lambda}\mathbb{C}[[\check{\lambda}]]$, its Borel transform $\mathcal{B}(a)(\xi)$, is given by:
\begin{equation}
    \mathcal{B}\colon\check{\lambda}\mathbb{C}[[\check{\lambda}]]\to \mathbb{C}[[\xi]], \;\;\;\; \mathcal{B}(\check{\lambda}^{n+1})=\frac{\xi^n}{n!}.
\end{equation}
Let $\check{\lambda}\in \mathbb{C}^{\times}$ and let $\rho$ be a ray from $0$ to $\infty$ in the complex $\xi$-plane. If $\mathcal{B}(a)(\xi)$ defines an analytic function along $\rho$, we define the Borel sum of $a(\check{\lambda})$ at $\check{\lambda}$, along $\rho$ by 
\begin{equation}\label{defborelsum}
    \int_{\rho}\mathrm d\xi \;e^{-\xi/\check{\lambda}}\mathcal{B}(a)(\xi)\,.
\end{equation}
If \eqref{defborelsum} is finite, we say $a(\check{\lambda})$ is Borel summable at $\check{\lambda}$, along $\rho$.

Writing
\begin{align*}\label{formal}
\widetilde{F}^{con}(\lambda,t) =\frac{1}{\lambda^2} \mathrm{Li}_{3}(Q)  +\frac{B_2}{2}\mathrm{Li}_1(Q)+ \Phi^{con}(\check{\lambda},t)\,, \quad \check{\lambda}=\frac{\lambda}{2\pi}\,, \quad Q=e^{2\pi \I t}\,,
\end{align*} 
the following was proven in \cite{Alim:2021mhp} for the Borel resummation of the series $\Phi(\check \lambda,t)$ :
\begin{itemize}
\item[(i)] For $t\in \mathbb{C}^{\times}$ with $|\mathrm{Re}(t)|<1/2$, let $G^{con}(\xi,t):=\mathcal{B}(\Phi(-,t))(\xi)$ denote the Borel transform of $\Phi(\check \lambda,t)$. Then $G^{con}(\xi,t)$ converges for $|\xi|<2\pi |t|$ and admits a series representation of the form
    \begin{equation}\label{Borel-sum}
        G^{con}(\xi,t) = \frac{1}{(2\pi )^2}\sum_{m\in\mathbb{Z} \setminus \{0\}}\frac{1}{m^3}
        \frac{1}{2\xi}\frac{\partial}{\partial \xi}
        \bigg(\frac{\xi^2}{1-e^{-2\pi \I t + \xi/m}}-\frac{\xi^2}{1-e^{-2\pi \I t - \xi/m}}\bigg),
\end{equation}
which can be used to analytically continue $G^{con}(\xi,t)$ in the $\xi$ variable to a meromorphic function with poles at $\xi=2\pi \I(t+k)m$ for $k \in \mathbb{Z}$ and $m \in \mathbb{Z}\setminus\{0\}$.
\item[(ii)] For $t\in \mathbb{C}-\mathbb{Z}$ and  $k\in \mathbb{Z}$ let $l_k:=\mathbb{R}_{<0}\cdot 2\pi \I(t+k)$ and $l_{\infty}:=\I\mathbb{R}_{<0}$. Given any ray $\rho$ from $0$ to $\infty$ different from $\{\pm l_{k}\}_{k\in \mathbb{Z}}\cup \{\pm l_{\infty}\}$, and $\lambda$ in the half-plane $\mathbb{H}_{\rho}$ centered at $\rho$, we define the Borel sum of $\widetilde{F}^{con}(\lambda,t)$ along $\rho$ as
\begin{equation}\label{BorelS}
    F^{con}_{\rho}(\lambda,t):=\frac{1}{\lambda^2} \mathrm{Li}_{3}(Q)+ \frac{B_2}{2}\mathrm{Li}_{1}(Q)+\int_{\rho}\mathrm d\xi\, e^{-\xi/\check{\lambda}} G^{con}(\xi,t)\,.
\end{equation}
\item[(iii)] Let $\rho_k$ be a ray in the sector determined by the Stokes rays $l_{k}$ and $l_{k-1}$. Then if $\mathrm{Im}(t)>0$, on the overlap of their domains of definition in the $\lambda$ variable we have
\begin{align}
    \phi_{\pm l_k}(\lambda,t)&:=F^{con}_{\pm \rho_{k+1}}(\lambda,t)
    -F^{con}_{\pm \rho_k}(\lambda,t) =\frac{1}{2\pi \I}\partial_{\check\lambda}\Big(\check{\lambda}\,\mathrm{Li}_2\big(e^{\pm 2\pi \I(t+k)/\check \lambda}\big)\Big)\;, \, \\
    &= \frac{1}{2\pi \I} \left( \mathrm{Li}_2\big(e^{\pm 2\pi \I(t+k)/\check \lambda}\big) \pm 2\pi \I \frac{(t+k)}{\check \lambda} \log \left(1- e^{\pm 2\pi \I(t+k)/\check \lambda} \right)\right)\;.
\end{align}
If $\mathrm{Im}(t)<0$, then the previous jumps also hold provided $\rho_{k+1}$ is interchanged with $\rho_{k}$ in the above formula. \\
\item[(iv)]   Let $\rho_k$ denote any ray between the rays $l_{k}$ and $l_{k-1}$. Furthermore, assume that $0<\mathrm{Re}(t)<1$, $\mathrm{Im}(t)>0$, $\mathrm{Re}(\lambda)>0$, $\mathrm{Im}(\lambda)<0$, and $\mathrm{Re}\, t < \mathrm{Re} (\check{\lambda}+1)$. Then
\begin{equation}
    \lim_{k\to \infty}F^{con}_{\rho_k}(\lambda,t)=\lim_{k\to \infty}F^{con}_{-\rho_k}(-\lambda,t)=\sum_{k=1}^\infty\frac{e^{2\pi \I k t}}{k\big(2\sin\big(\frac{\lambda k}{2}\big)\big)^2}\,.
\end{equation}
Furthermore, we can write the sum of the Stokes jumps along $l_k$ for $k\geq 0$ as
\begin{equation}
    \sum_{k=0}^{\infty}\phi_{l_k}(\lambda,t)=\frac{1}{2\pi \I }\partial_{\lambda}\Big(\lambda \sum_{m=1}^{\infty}\frac{Q'^l}{m^2(1-q'^m)}\Big), \;\;\;\; Q':=e^{2\pi \I t/\check\lambda},\;\; q':=e^{2\pi \I/\check{\lambda}}\;.
\end{equation}
If, on the other hand, we take $0<\mathrm{Re}(t)<1$, $\mathrm{Im}(t)>0$, $\mathrm{Re}(\lambda)>0$, $\mathrm{Im}(\lambda)>0$, $\mathrm{Re}(t)<\mathrm{Re}(\check\lambda +1)$ and furthermore assume that $|e^{2\pi \I t/\check{\lambda}}|<1$, then we also have
\begin{equation}
    \lim_{k\to -\infty}F^{con}_{\rho_k}(\lambda,t)=\lim_{k\to -\infty}F^{con}_{-\rho_k}(-\lambda,t)=\sum_{k=1}^\infty\frac{e^{2\pi \I k t}}{k\big(2\sin\big(\frac{\lambda k}{2}\big)\big)^2}.
\end{equation}
 \end{itemize}

$\lim_{k\to \infty}F^{con}_{\rho_k}(\lambda,t)$ differs from $F^{con}_{\mathbb{R}_{>0}}(\la,t)$ by the sum over all jumps 
$\phi_{l_k}(\lambda,t)$ for $k\geq 0$, leading to the decomposition, see \cite{Alim:2021ukq,Alim:2021mhp}:
\begin{equation}\label{Fnp-decomp}
F^{con}_{\BR_{>0}}(\la,t)= \sum_{k=1}^\infty\frac{e^{2\pi \I k t}}{k\big(2\sin\big(\frac{\lambda k}{2}\big)\big)^2}  -\frac{1}{2\pi \I }\partial_{\lambda}\Big(\lambda \sum_{m=1}^{\infty}\frac{Q'^l}{m^2(1-q'^m)}\Big).
\end{equation}

$F^{con}_{\mathbb{R}_{>0}}(\la,t)$ could furthermore be identified with the solution $\widetilde{F}^{con}_{\text{np}}(\lambda,t) $ of the difference equation \eqref{eq:diffeq}. 
\begin{equation}
\begin{split}
F^{con}_{\mathbb{R}_{>0}}(\la,t) =\widetilde{F}^{con}_{\text{np}}(\lambda,t) :=  \left(\frac{\pi \I}{6} \, B_{3,3}\left(t+\check\lambda\,|\,\check{\lambda},\check{\lambda},1 \right)\right) + \log\left(\sin_3\left(t+\check{\lambda} \,|\, \check{\lambda},\check{\lambda},1\right)\right).
\end{split}
\end{equation}

We proceed by giving the product structure of the partition functions obtained from the Borel sums corresponding to the free energies, these follow from the results of \cite{Alim:2021mhp} but were not contained in that paper.

\begin{prop}{Product form of the partition functions.}\label{prop:prod}\\
For the partition functions $$Z_{\rho}(\lambda,t):=\exp (F_{\rho}(\lambda,t))\,,$$ we have the following product forms:
\begin{align}
Z_{\rho_{\infty}} (\lambda,t)&= \prod_{m=0}^{\infty} (1- Q q^{m+1})^{m+1}\, , \\
 Z_{\rho_{0}} (\lambda,t) &= \prod_{m=0}^{\infty} (1-Q\, q^{m+1})^{m+1}  \cdot \exp\left(-\frac{1}{2\pi i} \operatorname{Li}_2(Q' q'^m)\right) \cdot (1-Q' q'^m)^{-\frac{2\pi}{\lambda}\left( t+m\right)}\,,\\
Z^{con}_{\rho_k}(\lambda,t)& = \prod_{m=0}^{\infty} (1-Q\, q^{m+1})^{m+1}  \cdot \prod_{m=k}^{\infty} \exp\left(-\frac{1}{2\pi i} \operatorname{Li}_2(Q' q'^m)\right) \cdot (1-Q' q'^m)^{-\frac{2\pi}{\lambda}\left( t+m\right)}\,,
\end{align}
where $Q=\exp(2\pi i t),\quad q=e^{\ii\lambda}\,, \quad Q'=\exp(2\pi i t/\check{\lambda})\,, \quad q'=\exp(2\pi i/\check{\lambda})$\,. For the two special rays $\rho_0$ and $\rho_{\infty}$ we will use special names for the corresponding partition functions:
$$  Z^{con}_{GV}(\lambda,t):= Z_{\rho_{\infty}} (\lambda,t) \, , \quad  Z^{con}_{np}(\lambda,t):= Z_{\rho_{0}} (\lambda,t)\,.$$
\end{prop}
\begin{proof}
To prove the first product form we use the following expansion:
\begin{equation}\label{sineexpansion}
\frac{1}{\left(2\sin\left(k\lambda/2\right)\right)^2}= -\sum_{m=1}^{\infty} m e^{\ii m k \lambda}\,,
\end{equation}
valid for $\textrm{Im}\lambda >0 $ and which can be obtained from the geometric series by a derivation in $\lambda$. We have 

\begin{equation}
F_{GV}(\lambda,t):=F_{\infty}(\lambda,t)= \sum_{k=1}^\infty\frac{e^{2\pi \I k t}}{k\big(2\sin\big(\frac{\lambda k}{2}\big)\big)^2}\,,
\end{equation}
which can be re-written using the identity \eqref{sineexpansion} as:
\begin{equation}
F_{GV}(\lambda,t)= - \sum_{m=1}^{\infty} m \sum_{k=1}^\infty\frac{e^{2\pi \I k t} e^{\ii m k \lambda}}{k} =  \sum_{m=1}^{\infty} m \,\log(1-e^{2\pi \I  t} e^{\ii m  \lambda}) \,,
\end{equation}
the product form then follows by taking $Z^{con}_{GV}(\lambda,t)=\exp(F_{GV}(\lambda,t))$ and shifting the summation index.
The product form for $Z_{\rho_0}$ follows from the form:
\begin{equation}
F_{\rho_0}(\la,t)= F^{con}_{\mathbb{R}>0}(\lambda,t)= \sum_{k=1}^\infty\frac{e^{2\pi \I k t}}{k\big(2\sin\big(\frac{\lambda k}{2}\big)\big)^2}  -\sum_{k=0}^{\infty}\phi_{l_k}(\lambda,t)\,,
\end{equation}
and the expressions for the Stokes jumps:
\begin{align}
    \phi_{ l_k}(\lambda,t)=\frac{1}{2\pi \I} \left( \mathrm{Li}_2\big(e^{ 2\pi \I(t+k)/\check \lambda}\big) + 2\pi \I \frac{(t+k)}{\check \lambda} \log \left(1- e^{ 2\pi \I(t+k)/\check \lambda} \right)\right)\;,
\end{align}
the product expression then follows by exponentiation, the product form for $Z_{\rho_k}$ follows from taking into account the exponentiated Stokes factors.

\end{proof}

\subsection{Constant maps contributions}\label{sec:const}

For the formal series giving the constant map contributions to the GW potential for any CY we have
\begin{align}\label{formal}
F_{c}(\lambda,t) = \sum_{g=0}^{\infty} F^g_c &= F_c^0 (t) + F_c^1(t)+  \chi(X)  \sum_{g=2}^{\infty} \frac{ (-1)^{g-1}\, B_{2g}\, B_{2g-2}}{4g (2g-2)\, (2g-2)!} \lambda^{2g-2} \,, \quad \check{\lambda}=\frac{\lambda}{2\pi}\,\,\\
&= F_c^0 (t) + F_c^1(t)+ \Phi_c(\check{\lambda})\,, \quad \check{\lambda}=\frac{\lambda}{2\pi}\,\,,
\end{align} 
the terms at genus $0$ and $1$ are $t$ dependent and encode also classical topological data of the CY such as the triple intersections and second Chern class, see \cite{pandharipande2003three} for details.

In \cite{Alim:2021mhp}, the case of the resolved conifold was treated for which $\chi(X)=2$ and the constant map contribution has the form:
\begin{align} \label{constmapcon}
F^{con}_{c}(\lambda,t)&= F_c^0 (t) + F_c^1(t)+   \sum_{g=2}^{\infty} \frac{ (-1)^{g-1}\, B_{2g}\, B_{2g-2}}{2g (2g-2)\, (2g-2)!} \lambda^{2g-2} \,, \quad \check{\lambda}=\frac{\lambda}{2\pi}\,\\
&= F_c^0 (t) + F_c^1(t)+  \Phi_c^{con}(\check{\lambda})\,.
\end{align} 

It was shown in \cite{Alim:2021mhp} that the Borel transform of $ \Phi_c^{con}(\check{\lambda})$ is given by:
\begin{equation}
G_c^{con}(\xi)= -G^{con}(\xi,0)-\frac{1}{12\xi}\,,
\end{equation}
and that correspondingly the Borel sums of the constant map contributions are given by
\begin{equation}
F^{con}_{c,\rho}(\lambda)= - F^{con}_{\rho}(\lambda,0)-\frac{1}{12}\log \check{\lambda}+C\,,
\end{equation}
with $C$ a constant which does not play a role. 
We thus obtain that the Borel transform of $\Phi_c(\check{\lambda})$ in the more general case is given by:
    \begin{equation}\label{Borel-sum-const}
        G_c(\xi) =-\frac{\chi(X)}{24\xi} - \frac{\chi(X)}{2(2\pi )^2}\sum_{m\in\mathbb{Z} \setminus \{0\}}\frac{1}{m^3}
        \frac{1}{2\xi}\frac{\partial}{\partial \xi}
        \bigg(\frac{\xi^2}{1-e^{ \xi/m}}-\frac{\xi^2}{1-e^{- \xi/m}}\bigg),
\end{equation}
which can be analytically continued in the $\xi$ variable to a meromorphic function with poles at $\xi=2\pi \I k $ for $k \in \mathbb{Z}$ .

Furthermore given any ray $\rho$ from $0$ to $\infty$ different from $\{\pm l_{\infty}\}$, and $\lambda$ in the half-plane $\mathbb{H}_{\rho}$ centered at $\rho$, the Borel sum of $F_c(\lambda)$ along $\rho$ is given by
\begin{equation}\label{BorelSconst}
    F_{c,\rho}(\lambda):=\frac{1}{\lambda^2} F^0_c(t)+ F^1_c(t) +\int_{\rho}\mathrm d\xi\, e^{-\xi/\check{\lambda}} G_c(\xi)\,.
\end{equation}

\begin{prop}{Product form of the partition function.}\label{prop:prodconst}\\
For the partition function $$ Z^{con}_{c,np}(\lambda):= \exp (F^{con}_{c,\rho}(\lambda)) \check{\lambda}^{-1/12}e^C $$ 

we have the following product form:
\begin{align}
Z^{con}_{c,np}(\lambda) =\left(e^C \check{\lambda}^{-1/12} \exp\left(\frac{\ii \pi}{12} \right) \prod_{m=0}^{\infty} (1- q^{m+1})^{m+1}  \cdot \exp\left(-\frac{1}{2\pi i} \operatorname{Li}_2(q'^m)\right) \cdot (1-q'^m)^{-\frac{ 2\pi m}{\lambda}}\right)^{-1}\,,
\end{align}
where  $ q= \exp( i \lambda) \quad q'=\exp(2\pi i/\check{\lambda})$\,,
\end{prop}
\begin{proof}
This follows from evaluating the $t\rightarrow 0$ limit of the conifold partition function.
\end{proof}
\begin{rem}
We note that $$ \prod_{m=0}^{\infty} (1- q^{m+1})^{-m-1}$$ corresponds to the MacMahon function and that $Z^{con}_{c,np}(\lambda)$ is a non-perturbative completion thereof which is likely interesting in its own right, this has potentially an interpretation as a quantum modular form \cite{zagier2010quantum}.
\end{rem}
For the constant map contributions for Gromov-Witten theory on an arbitrary CY manifold $X$, we introduce the following non-perturbative generating function which is valid up to the genus $0$ and $1$ terms which should be included separately.

\begin{equation}
Z_{c,np}(\lambda)= \left( \prod_{m=0}^{\infty} (1- q^{m+1})^{m+1}  \cdot \exp\left(-\frac{1}{2\pi i} \operatorname{Li}_2(q'^m)\right) \cdot (1-q'^m)^{-\frac{ 2\pi m}{\lambda}}\right)^{-\chi(X)/2}\,.
\end{equation}

\section{Non-perturbative topological strings}\label{sec:nonpert}

\subsection{Sheaf invariants}\label{subsec:sheafinv}
Our first goal is to re-write the topological string free energy on an arbitrary CY threefold in terms of the free energy of the resolved conifold. This will be possible using a re-arrangement of the GV invariants in terms of a set of sheaf invariants $\Omega_{\beta,n}$  defined in the following way:
\begin{equation}\label{defsheafinv}
\sum_{n\in \mathbb{Z}\,} \Omega_{\beta,n} y^n = \sum_{g\ge0} [\textrm{GV}]_{\beta,g}\, (-1)^g (y^{1/2}-y^{-1/2})^{2g}\,,
\end{equation}
since the RHSis invariant under $y\rightarrow y^{-1}$ we have:
\begin{equation}\label{reflection}
\Omega_{\beta,n}=\Omega_{\beta,-n}\,, \forall n \in \mathbb{Z}\,,
\end{equation}
these are up to a sign the sheaf invariants put forward by Maulik and Toda \cite{MaulikToda}, following works of Hosono-Saito-Takahashi \cite{HST} and Kiem-Li \cite{KiemLi} to define the Gopakumar-Vafa invariants in terms of moduli spaces of sheaves, we are using the notation of \cite{Stoppa} which should also be consulted for further details and references.

It is useful to note the relation of how each $\Omega_{\beta,k}\, $ can be obtained from the GV invariants. We recall the finiteness of the GV invariants discussed in \ref{sec:finiteness}, stating that for every curve class $\beta$, there exists a $g_{\beta}\in \mathbb{N}$ such that 
$$[\textrm{GV}]_{\beta,g'}=0 \,, \quad \forall g'\ge g_{\beta}\,,$$
by comparing the $y^k$ coefficient on both sides of \eqref{defsheafinv} we obtain:
\begin{equation}
\Omega_{\beta,k}=  \left\{
	\begin{array}{ll}
		(-1)^k \sum_{g=k}^{g_{\beta}-1}   [\textrm{GV}]_{\beta,g} \binom{2g}{g-k} \,, \quad k\ge 0\,,   & \mbox{if } k < g_{\beta} \\
		0 & \mbox{if }  k \ge g_{\beta}
	\end{array}
\right.\end{equation}
the $k<0$ invariants are obtained by \eqref{reflection}. 

We further note that accordingly the sum $\sum_{n\in\mathbb{Z}}\Omega_{\beta,n}$ is finite. This sum is equal to the limit $y\rightarrow 1$ of the LHS of
\begin{equation*}
\sum_{n\in \mathbb{Z}\,} \Omega_{\beta,n} y^n = \sum_{g\ge0} [\textrm{GV}]_{\beta,g}\, (-1)^g (y^{1/2}-y^{-1/2})^{2g}\,,
\end{equation*}
evaluating the limit on the RHS we obtain:
\begin{equation}\label{eq:sumOmega}
\sum_{n\in \mathbb{Z}} \Omega_{\beta,n}  = [\textrm{GV}]_{\beta,0}\,.
\end{equation}

\subsection{Free energy in terms of the resolved conifold}

\begin{thm}\label{thm:conifolds}
The generating function of GW invariants, excluding the classical terms as well as the contributions from constant maps for any Calabi-Yau threefold:
\begin{equation}
\tilde{F}_{GW}(\lambda,t)= \sum_{g=0}^{\infty}\sum_{\beta>0}  \lambda^{2g-2} [GW]_{\beta,g} Q^{\beta}\,,
\end{equation} 
can be written as:
\begin{align}
\tilde{F}_{GW}(\lambda,t)&=  \sum_{n\in \mathbb{Z}}  \sum_{\beta>0}   \Omega_{\beta,n} \widetilde{F}^{con} (\lambda, t^{\beta,n}) \, \\
&=\sum_{n\in \mathbb{Z}}  \sum_{\beta>0}   \Omega_{\beta,n}  \left( \frac{1}{\lambda^2} \mathrm{Li}_{3}(Q^{\beta,n})+ \sum_{g=1}^{\infty} \lambda^{2g-2} \frac{(-1)^{g-1}B_{2g}}{2g (2g-2)!} \mathrm{Li}_{3-2g} (Q^{\beta,n}) \, \right)\,,
\end{align}
where $t^{\beta,n}:=t^{\beta}+ n \check{\lambda}\,,Q^{\beta,n}:=\exp(2\pi \ii t^{\beta,n})$ and $\Omega_{\beta,n}$ are the integers defined in \eqref{defsheafinv}. A similar relation holds for the generating function of GV invariants:
\begin{equation}
F_{GV}(\lambda,t)  = \sum_{\beta>0}\sum_{g\ge 0} [\textrm{GV}]_{\beta,g}\, \sum_{k\ge 1} \frac{1}{k} \left( 2 \sin \left( \frac{k\lambda}{2}\right)\right)^{2g-2} Q^{k\beta}\,,
\end{equation}
namely:
\begin{align}
F_{GV}(\lambda,t)  &= \sum_{n\in \mathbb{Z}}  \sum_{\beta>0}   \Omega_{\beta,n}  \sum_{m=1}^{\infty} \log\left(1-   Q^{\beta}   \exp(\ii (n+m) \,\lambda) \right)^m  \,,\\
&=  \sum_{n\in \mathbb{Z}}  \sum_{\beta>0}   \Omega_{\beta,n} F_{GV}^{con} (\lambda, t^{\beta}+ n \check{\lambda}) \,.
\end{align}

\end{thm}

\begin{proof}

We start with the Gopakumar-Vafa form of the topological string free energy:

\begin{equation}
F_{GV}(\lambda,t)= \sum_{\beta>0} \sum_{k\ge 1}  \frac{Q^{k\beta}}{k \left( 2 \sin \left( \frac{k\lambda}{2}\right)\right)^{2}}  \sum_{g\ge 0} [\textrm{GV}]_{\beta,g}\, \left( 2 \sin \left( \frac{k\lambda}{2}\right)\right)^{2g} \,,
\end{equation}

using \eqref{defsheafinv} with $y=\exp(\ii k\lambda)$ this becomes:

\begin{equation}\label{eq:GVinTermsofOmega}
F_{GV}(\lambda,t)= \sum_{\beta>0} \sum_{k\ge 1}  \frac{Q^{k\beta}}{k \left( 2 \sin \left( \frac{k\lambda}{2}\right)\right)^{2}} \sum_{n\in \mathbb{Z}\,} \Omega_{\beta,n} \exp(\ii  n k\lambda) \,,
\end{equation}

using the following identity which can be derived from the geometric series:

\begin{equation}
\frac{1}{\left( 2\sin \left(k\lambda/2\right)\right)^2} = -\sum_{m=1}^{\infty} m \exp(\ii m k \lambda)\,,
\end{equation}
we obtain:
\begin{equation}
F_{GV}(\lambda,t)=- \sum_{m=1}^{\infty} \sum_{n\in \mathbb{Z}} \sum_{\beta>0} \sum_{k\ge 1}  \frac{m\, Q^{k\beta}}{k} \Omega_{\beta,n}  \exp(\ii (n+m) k\,\lambda) \,,
\end{equation}

we can further re-write this as:

\begin{equation}
F_{GV}(\lambda,t)=  \sum_{n\in \mathbb{Z}}  \sum_{\beta>0}   \Omega_{\beta,n}  \sum_{m=1}^{\infty} \log\left(1-   Q^{\beta}   \exp(\ii (n+m) \,\lambda) \right)^m  \,,
\end{equation}

which can be expressed in terms of the free energy of the resolved conifold as:
\begin{align}
F_{GV}(\lambda,t)=  \sum_{n\in \mathbb{Z}}  \sum_{\beta>0}   \Omega_{\beta,n} F_{GV}^{con} (\lambda, t^{\beta}+ n \check{\lambda}) \, .
\end{align}
In this last expression we introduce $t^{\beta,n}:=  t^{\beta}+ n \check{\lambda}$ and consider this as a formal variable and look at the asymptotic expansion of each summand in $\lambda$ keeping $t^{\beta,n}$ fixed, obtaining  the following formal series, which is asymptotic in $\lambda\rightarrow 0$:
\begin{equation}
\tilde{F}_{GW}(\lambda,t)=  \sum_{n\in \mathbb{Z}}  \sum_{\beta>0}   \Omega_{\beta,n}  \left( \frac{1}{\lambda^2} \mathrm{Li}_{3}(Q^{\beta,n})+ \sum_{g=1}^{\infty} \lambda^{2g-2} \frac{(-1)^{g-1}B_{2g}}{2g (2g-2)!} \mathrm{Li}_{3-2g} (Q^{\beta,n}) \, \right)\,,
\end{equation}
where $Q^{\beta,n}:=\exp(2\pi \ii t^{\beta,n}).$

\end{proof}

\subsection{Non-perturbative free energy}

Having expressed the generating function of GW invariants in an arbitrary CY as a sum over generating functions of the GW invariants of the resolved conifold, it is now natural to use the previously obtained analytic functions in the topological string coupling to put forward a non-perturbative expression for the topological string free energy in the holomorphic limit.

The expression of the asymptotic series of GW invariants of an arbitrary CY threefold in terms of the resolved conifold:
\begin{equation}
\tilde{F}_{GW}(\lambda,t)= \sum_{n\in \mathbb{Z}}  \sum_{\beta>0}   \Omega_{\beta,n} F_{GW}^{con} (\lambda, t^{\beta}+ n \check{\lambda}) \,,
\end{equation}

suggests the use of the resurgence analysis of the resolved conifold to define the following:
\begin{align}
\widetilde{F}_{top,np}(\lambda,t) &= \sum_{n\in \mathbb{Z}}  \sum_{\beta>0}   \Omega_{\beta,n} F^{con}_{np} (\lambda, t^{\beta}+ n \check{\lambda}) \, \\
&= \sum_{n\in \mathbb{Z}}  \sum_{\beta>0}   \Omega_{\beta,n} \log\left( \mathcal{S}_3(t+ n \check{\lambda}| \check{\lambda},1)\right)\,.
\end{align}
On the RHS of this equation each summand has good analytic properties in $\lambda$ from the triple sine function, to be able to say whether this sum also defines a good analytic function one must however be careful with the summation coefficients, which are, especially for compact Calabi-Yau manifolds generically unknown. We note however that for each curve class $\beta$ the summation of the coefficients $\Omega_{\beta,n}$ is given by \eqref{eq:sumOmega}:
\begin{equation}
\sum_{n\in \mathbb{Z}} \Omega_{\beta,n}  = [\textrm{GV}]_{\beta,0}\,,
\end{equation}
although now there are generically infinitely many $ [\textrm{GV}]_{\beta,0}$ which are non-vanishing for compact CY manifolds, it is known by mirror symmetry that all these invariants are generated by rational functions on the moduli space of the mirror manifolds, the Yukawa couplings, see \cite{CoxKatz} for more details. It seems thus reasonable to expect that the proposed non-perturbative expression could correspond to an actual function. We will later see however that although the issue with potential growth of the $\Omega_{\beta,n}$ remains the same as for the perturbative expansion, the non-perturbative corrections obtained in this way nicely assemble into a function which will be given in sec. \ref{sec:defprepot}.

We then have the following:
\begin{prop}
Define $$\tilde{Z}_{top,np}(\lambda,t)= \exp\left( \tilde{F}_{top,np}(\lambda,t) \right)$$
then we have the following product expansion:
\begin{align}
\tilde{Z}_{top,np}(\lambda,t)&:=  Z_{top,wk}(\lambda,t) \times Z_{top,sg}(\lambda,t)
\end{align}
where
\begin{align}
 Z_{top,wk}(\lambda,t)&:= \prod_{n\in \mathbb{Z}}  \prod_{\beta>0}  \prod_{m=1}^{\infty}   \left(1-   Q^{\beta}   \exp(\ii (n+m) \,\lambda) \right)^{m\, \Omega_{\beta,n}}\,,\\
 Z_{top,sg}(\lambda,t) &:=   \prod_{\beta>0} \exp\left(-\frac{[\textrm{GV}]_{\beta,0}}{2\pi i} \operatorname{Li}_2(Q'^{\beta} q'^m)\right) \cdot (1-Q'^{\beta} q'^m)^{-\frac{ [\textrm{GV}]_{\beta,0}}{\check{\lambda}}\left( t^{\beta}+m\right)}\,,
\end{align}
$t^{\beta,n}:=t^{\beta}+ n \check{\lambda}\,, t^{\beta,0}:=t^{\beta}\, , Q^{\beta,n}:=\exp(2\pi \ii t^{\beta,n})\,, Q^{\beta}:=Q^{\beta,0}$
and $q=e^{\ii\lambda}\,, \quad Q'=\exp(2\pi i t/\check{\lambda})\,, \quad q'=\exp(2\pi i/\check{\lambda})$\,. The subscripts $wk$ and $sg$ stand for weak and strong and refer to the pieces of the partition function which are expansions in $q$ or $q'$ respectively.

\end{prop}

\begin{proof}
From the definition
\begin{align}
\tilde{F}_{top,np}(\lambda,t) &= \sum_{n\in \mathbb{Z}}  \sum_{\beta>0}   \Omega_{\beta,n} F^{con}_{np} (\lambda, t^{\beta}+ n \check{\lambda}) \,, 
\end{align}
we obtain:
\begin{align}
Z_{top,np}(\lambda,t) = \prod_{n\in \mathbb{Z}}  \prod_{\beta>0}    Z^{con}_{np} (\lambda, t^{\beta}+ n \check{\lambda})^{\Omega_{\beta,n}} \,, 
\end{align}
we can now use the product expansion of \eqref{prop:prod}:

$$Z^{con}_{np}(\lambda,t)= \prod_{m=0}^{\infty} (1-Q\, q^{m+1})^{m+1}  \cdot \exp\left(-\frac{1}{2\pi i} \operatorname{Li}_2(Q' q'^m)\right) \cdot (1-Q' q'^m)^{-\frac{2\pi}{\lambda}\left( t+m\right)}\,, $$

to obtain:
\begin{align}
Z_{top,np}(\lambda,t) = \prod_{n\in \mathbb{Z}}  \prod_{\beta>0}   \left( \prod_{m=0}^{\infty} (1-Q^{\beta,n}\, q^{m+1})^{m+1}  \cdot \exp\left(-\frac{1}{2\pi i} \operatorname{Li}_2(Q'^{\beta} q'^m)\right) \cdot (1-Q'^{\beta} q'^m)^{-\frac{2\pi}{\lambda}\left( t^{\beta,n}+m\right)}\,\right) ^{\Omega_{\beta,n}} \,, 
\end{align}
where $t^{\beta,n}:=t^{\beta}+ n \check{\lambda}\,, t^{\beta,0}:=t^{\beta}\, , Q^{\beta,n}:=\exp(2\pi \ii t^{\beta,n})\,, Q^{\beta}:=Q^{\beta,0}$ as before, and the first surprise is that the primed variables are independent of $n$ since:
$$ Q'^{\beta,n}:=\exp(2\pi \ii t^{\beta,n}/\check{\lambda})=\exp(2\pi \ii t^{\beta}/\check{\lambda})\cdot \exp(2\pi \ii n)=Q'^{\beta}\,, \quad \forall n\in \mathbb{Z}\,,$$
we define
\begin{equation}
Z_{np,sg}(\lambda,t):=\prod_{n\in \mathbb{Z}}  \prod_{\beta>0} \prod_{m=0}^{\infty}  \left(\exp\left(-\frac{1}{2\pi i} \operatorname{Li}_2(Q'^{\beta} q'^m)\right) \cdot (1-Q'^{\beta} q'^m)^{-\frac{2\pi}{\lambda}\left( t^{\beta,n}+m\right)}\,\right)^{\Omega_{\beta,n}} \,, 
\end{equation}
and continue to show that it is of the desired form by first writing:
\begin{equation}
Z_{np,sg}(\lambda,t):=\prod_{\beta>0}  \prod_{m=0}^{\infty}   \left(\exp\left(-\frac{1}{2\pi i} \sum_{n\in\mathbb{Z}} \Omega_{\beta,n}\, \operatorname{Li}_2(Q'^{\beta} q'^m)\right) \cdot (1-Q'^{\beta} q'^m)^{-\frac{2\pi}{\lambda} \sum_{n\in\mathbb{Z}} \Omega_{\beta,n} \left( t^{\beta,n}+m\right)}\,\right) \,, 
\end{equation}
noting that the argument of the dilogarithm does not depend on $n$, the prefactor becomes $\sum_{n\in\mathbb{Z}} \Omega_{\beta,n}$. This sum is equal equal to \eqref{eq:sumOmega}\begin{equation}
\sum_{n\in \mathbb{Z}} \Omega_{\beta,n}  = [\textrm{GV}]_{\beta,0}\,.
\end{equation}
We now consider the exponent of the second factor:

\begin{align}
-\frac{1 }{\check{\lambda}} \sum_{n\in\mathbb{Z}} \Omega_{\beta,n} \left( t^{\beta,n}+m\right) &=-\frac{1 }{\check{\lambda}} \sum_{n\in\mathbb{Z}} \Omega_{\beta,n} \left( t^{\beta}+  n\check{\lambda} +m\right)\\ &=-\frac{1 }{\check{\lambda}} [\textrm{GV}]_{\beta,0}\,  \left( t^{\beta}+   m\right)  -  \sum_{n\in\mathbb{Z}} n\,\Omega_{\beta,n} \\ 
&=-\frac{1 }{\check{\lambda}} [\textrm{GV}]_{\beta,0}\,  \left( t^{\beta}+   m\right)   \,.
\end{align}
the summand
$ \sum_{n\in\mathbb{Z}} n\,\Omega_{\beta,n}=0$ since $\Omega_{\beta,n}=\Omega_{\beta,-n}\, \forall n\in \mathbb{Z}$. These cancellations yield the desired form:
\begin{align}
 Z_{top,sg}(\lambda,t) =   \prod_{\beta>0} \exp\left(-\frac{[\textrm{GV}]_{\beta,0}}{2\pi i} \operatorname{Li}_2(Q'^{\beta} q'^m)\right) \cdot (1-Q'^{\beta} q'^m)^{-\frac{ [\textrm{GV}]_{\beta,0}}{\check{\lambda}}\left( t^{\beta}+m\right)}\,.
\end{align}
\end{proof}
We pause to comment on this result, it seems rather remarkable that, due to some non-trivial cancellation only the genus zero GV invariants appear in the non-perturbative corrections to the topological string partition function. We will now proceed to show that these non-perturbative corrections can be cast in a rather simple function which we introduce as a deformation of the genus zero prepotential.

\subsection{Deformation of the prepotential}\label{sec:defprepot}
Recall that the non-classical part of the prepotential of a Calabi-Yau threefold has an expression in terms of GV invariants given by:

\begin{equation}
 \tilde{F}^0(t)=\sum_{\beta>0} [\textrm{GV}]_{\beta,0} \operatorname{Li}_3(Q^{\beta}) = \sum_{\beta>0}\sum_{k=1}^{\infty} \frac{[\textrm{GV}]_{\beta,0}}{k^3} Q^{k\beta} \,, \quad Q^{\beta}= \exp(2\pi i t^{\beta})\,.
 \end{equation}
We introduce the following $\varepsilon-$deformation of the prepotential:
\begin{equation}
\tilde{F}^{0}_{\textrm{def}}(\varepsilon,t):=\frac{1}{2} \sum_{\beta>0} \sum_{k=1}^{\infty} \frac{[\textrm{GV}]_{\beta,0}}{k^2} \frac{Q^{k\beta}}{\sin(k\varepsilon/2)}\,,
\end{equation}
this is chosen such that
\begin{equation}
\lim_{\varepsilon\rightarrow 0} \varepsilon \, \tilde{F}^{0}_{\textrm{def}}(\varepsilon,t) =   \tilde{F}^{0}(t)\,.
\end{equation}

\begin{rem}\label{rem:fdef}
\begin{itemize}
\item For the case of the resolved conifold with only one curve class $\beta$ and a nonvanishing GV invariant only at degree $1$, this deformed prepotential is:
\begin{equation}
\tilde{F}^{0}_{\textrm{def}}(\varepsilon,t):=\frac{1}{2}  \sum_{k=1}^{\infty} \frac{1}{k^2} \frac{Q^{k\beta}}{\sin(k\varepsilon/2)}= -F^{NS,con}(\varepsilon,t)\,,
\end{equation}
where $F^{NS,con}(\varepsilon,t)$ denotes the refined topological string in the Nekrasov-Shatashvili limit in the normalization of \cite{Alim:2022oll}. 
\item Note that here the deformation does not require a refined topological string origin and can therefore be applied to both non-compact as well as compact Calabi-Yau geometries.
\item In \cite{Bousseau:Anomaly}, the authors introduce in sec.~1.7  a definition for the topological string free energy in the NS limit for the geometry $\mathcal{O}(-3) \rightarrow \mathbb{P}^2$ known as local $\mathbb{P}^2$ as follows:
\begin{equation} \label{Bousseau:FNS}
F^{NS}:= i \sum_{d,k=1}^{\infty} \frac{1}{k^2} \frac{\Omega_d(y^{\frac{k}{2}})}{y^{\frac{k}{2}}-y^{\frac{-k}{2}}} Q^{kd}\,, 
\end{equation}
where $\Omega_d(y^{\frac{k}{2}})$ denote refined DT invariants studied in \cite{Bousseau:Scattering}. It is furthermore stated in \cite{Bousseau:Anomaly} that the refined invariants $\Omega_d(y^{\frac{1}{2}})$ become the genus zero GV invariants in the limit $y\rightarrow 1$. In that case the enumerative NS generating function coincides with the deformation of the prepotential introduced above with $y=e^{\ii \lambda}$.\footnote{In the formula for  \eqref{Bousseau:FNS} there is a typo in the $\frac{1}{k^2}$ which is given as $\frac{1}{k}$ in the arxiv preprint of that paper, I would like to thank Pierrick Bousseau for pointing this out to me.}
\end{itemize}
\end{rem}

 We can now express the non-perturbative corrections to the topological string as follows:
\begin{prop}\label{prop:correc}
The non-perturbative corrections to the topological string 
\begin{align} \label{Ztopsg}
 Z_{top,sg}(\lambda,t) :=   \prod_{\beta>0}  \prod_{m=0}^{\infty}  \exp\left(-\frac{[\textrm{GV}]_{\beta,0}}{2\pi i} \operatorname{Li}_2(Q'^{\beta} q'^m)\right) \cdot (1-Q'^{\beta} q'^m)^{-\frac{ [\textrm{GV}]_{\beta,0}}{\check{\lambda}}\left( t^{\beta}+m\right)}\,,
\end{align}
can be expressed as 
\begin{equation}
 Z_{top,sg}(\lambda,t)= \exp(F_{top,sg}(\lambda,t) )\,,
\end{equation}
where

\begin{equation}
F_{top,sg}(\lambda,t)=  -\frac{1}{2\pi} \partial_{\lambda} \left( \lambda \, \tilde{F}^{0}_{\textrm{def}} \left( \frac{2\pi}{\check{\lambda}}, \frac{t-1/2}{\check{\lambda}} \right)\right)\,,
\end{equation}

where 
\begin{equation}
\tilde{F}^{0}_{\textrm{def}}(\varepsilon,t):=\frac{1}{2} \sum_{k,\beta>0} \frac{[\textrm{GV}]_{\beta,0}}{k^2} \frac{Q^{k\beta}}{\sin(k\varepsilon/2)}\,.
\end{equation}

\end{prop}

\begin{proof}
From $
 Z_{top,sg}(\lambda,t)= \exp(F_{top,sg}(\lambda,t) )
$ and \eqref{Ztopsg} we obtain
 
\begin{align}
 F_{top,sg}(\lambda,t)=  -\sum_{\beta>0}  \sum_{m=0}^{\infty} \frac{[\textrm{GV}]_{\beta,0}}{2\pi i}  \left( \operatorname{Li}_2(Q'^{\beta} q'^m) +  \frac{ 2\pi \ii}{\check{\lambda}}\left( t^{\beta}+m\right) \,\log (1-Q'^{\beta} q'^m)\right) \,,
\end{align}
which can be written as
\begin{align}
 F_{top,sg}(\lambda,t)=  - \sum_{\beta>0}  \sum_{m=0}^{\infty} \frac{[\textrm{GV}]_{\beta,0}}{2\pi i} \partial_{\check{\lambda}} \left( \check{\lambda} \, \operatorname{Li}_2(e^{2\pi \ii (t^{\beta}+m)/\check{\lambda}})\right) \,,
\end{align}
writing out the series expansion for the dilogarithm, we obtain
\begin{align}
 F_{top,sg}(\lambda,t)=  - \sum_{\beta>0}  \sum_{m=0}^{\infty} \sum_{k=1}^{\infty} \frac{[\textrm{GV}]_{\beta,0}}{2\pi i} \partial_{\check{\lambda}} \left( \check{\lambda} \, \frac{e^{2\pi \ii (t^{\beta}+m)k/\check{\lambda}}}{k^2}\right) \,,
\end{align}
summing over $m$ using the geometric series gives
\begin{align}
 F_{top,sg}(\lambda,t)&=  -  \partial_{\check{\lambda}}  \left( \check{\lambda} \sum_{\beta>0}   \sum_{k=1}^{\infty} \frac{[\textrm{GV}]_{\beta,0}}{2\pi i}  \, \frac{e^{2\pi \ii t^{\beta}k/\check{\lambda}}}{k^2 (1-e^{2\pi \ii k/\check{\lambda}})} \right) \,\\
 &=    \partial_{\check{\lambda}}  \left( \check{\lambda} \sum_{\beta>0}   \sum_{k=1}^{\infty} \frac{[\textrm{GV}]_{\beta,0}}{2\pi i}  \, \frac{e^{2\pi \ii t^{\beta}k/\check{\lambda}}e^{-\pi \ii k/\check{\lambda}}}{k^2 (e^{\pi \ii k/\check{\lambda}}-e^{-\pi \ii k/\check{\lambda}})} \right) \,\\
  &=   -\frac{1}{2\pi} \partial_{\check{\lambda}}  \left( \check{\lambda} \frac{1}{2} \sum_{\beta>0}   \sum_{k=1}^{\infty} \frac{[\textrm{GV}]_{\beta,0}}{2\pi i}  \, \frac{e^{2\pi \ii (t^{\beta}-1/2)k/\check{\lambda}}}{k^2 \, \sin( \pi  k/\check{\lambda})} \right) \,\\&= -\frac{1}{2\pi} \partial_{\lambda} \left( \lambda \, \tilde{F}^{0}_{\textrm{def}} \left( \frac{2\pi}{\check{\lambda}}, \frac{t-1/2}{\check{\lambda}} \right)\right)\,.
\end{align}

\end{proof}

\subsection{Non-perturbative topological string partition function}

We can now include the contributions from constant maps to the topological string free energies and use the analysis of sec.~\ref{sec:buildingblocks} to obtain the following:
\begin{prop}
The non-perturbative partition function of topological string theory on any CY threefold, up to the contributions from constant maps at genus $0,1$, can be expressed as:

\begin{equation}
Z_{\textrm{top},np}(\lambda,t) =  Z_{c,np}(\lambda)  \prod_{n\in\mathbb{Z}} \prod_{\beta>0} \left( Z^{\textrm{con}}_{np} (\lambda, t^{\beta}+ n \, \check{\lambda})\right)^{\Omega_{\beta,n}}\,,
\end{equation}
where

$$Z^{con}_{np}(\lambda,t)= \prod_{m=0}^{\infty} (1-Q\, q^{m+1})^{m+1}  \cdot \exp\left(-\frac{1}{2\pi i} \operatorname{Li}_2(Q' q'^m)\right) \cdot (1-Q' q'^m)^{-\frac{2\pi}{\lambda}\left( t+m\right)}\,, $$

$$
Z_{c,np}(\lambda)= \left( \prod_{m=1}^{\infty} (1- q^{m})^{m}  \cdot \exp\left(-\frac{1}{2\pi i} \operatorname{Li}_2(q'^m)\right) \cdot (1-q'^m)^{-\frac{ m}{\check{\lambda}}}\right)^{-\chi(X)/2}
$$

\end{prop}

\begin{proof} 
This follows from the definition of the free energies by exponentiation as well as using the product form \eqref{prop:prod} and the analysis of sec.~\ref{sec:const}.
 
\end{proof}


\section{Resurgence}

We now proceed to study the analytic structure behind the asymptotic expansion of GW theory by using resurgence. We therefore first re-write the asymptotic expansion of the GW invariants in terms of the sheaf invariants. 

\subsection{Asymptotic series in terms of sheaf invariants}

\begin{prop}\label{prop:GWinsheaf}
The asymptotic series giving the non-constant part of the generating function of GW invariants 
\begin{align}\label{eq:FGWasmpt}
F_{GW}(\lambda,t) = \frac{1}{\lambda^2}\tilde{F}^0(t) + \tilde{F}^{1}(t) + \sum_{g=2}^{\infty}  \lambda^{2g-2} \tilde{F}^g(t) \,,
\end{align}

can be written in terms of the sheaf invariants as:

$$ \tilde{F}^0(t) =  \sum_{\beta>0} \sum_{n \in \mathbb{Z}} \Omega_{\beta,n}   \sum_{k\ge 1}   \frac{Q^{k\beta}}{k^3} \,, $$

\begin{align}
\tilde{F}^1(t)= \sum_{\beta>0} \sum_{n\in\mathbb{Z}}  \Omega_{\beta,n}     \frac{1}{2}   \left( n\, 2\,  B_{1}(n)-  B_{2}(n)\right)    \operatorname{Li}_{1}(Q^{\beta})  \,,
\end{align}

and 
\begin{equation}
\tilde{F}^g(t)= \sum_{\beta>0} \sum_{n\in\mathbb{Z}}\frac{(-1)^g \Omega_{\beta,n}}{2\cdot2g!} \left( (2g-1) (B_{2g}(n)+B_{2g}(-n))- 2g n (B_{2g-1}(n)-B_{2g-1}(-n))\right)  \operatorname{Li}_{3-2g}(Q^{\beta}) \,,
\end{equation}

for $g\ge 2$, where $B_m(x)$ denotes Bernoulli polynomials.

\end{prop}

\begin{proof}
To obtain a new expression for the Gromov-Witten potential in terms of the sheaf invariants we start again with the Gopakumar-Vafa form of the topological string free energy:

\begin{equation}
F_{GV}(\lambda,t)= \sum_{\beta>0} \sum_{k\ge 1}  \frac{Q^{k\beta}}{k \left( 2 \sin \left( \frac{k\lambda}{2}\right)\right)^{2}}  \sum_{g\ge 0} [\textrm{GV}]_{\beta,g}\, \left( 2 \sin \left( \frac{k\lambda}{2}\right)\right)^{2g} \,,
\end{equation}

using \eqref{defsheafinv} with $y=\exp(\ii k\lambda)$ this becomes:

\begin{align}
F_{GV}(\lambda,t) &= \sum_{\beta>0} \sum_{k\ge 1}  \frac{Q^{k\beta}}{k \left( 2 \sin \left( \frac{k\lambda}{2}\right)\right)^{2}} \sum_{n\in \mathbb{Z}\,} \Omega_{\beta,n} \exp(\ii  n k\lambda) \,,\\
&= \sum_{\beta>0} \sum_{n\in\mathbb{Z}}  \Omega_{\beta,n}   \sum_{k\ge 1}   \frac{Q^{k\beta}}{k}  \,  \frac{ \exp(\ii  n k\lambda)}{{\left( 2 \sin \left( \frac{k\lambda}{2}\right)\right)^{2}}} \,,\\
&=  \sum_{\beta>0} \sum_{n\in \mathbb{Z}} \frac{ \Omega_{\beta,n}}{2}   \sum_{k\ge 1}   \frac{Q^{k\beta}}{k}  \,  \frac{ \exp(\ii  n k\lambda+ \exp(-\ii  n k\lambda)}{{\left( 2 \sin \left( \frac{k\lambda}{2}\right)\right)^{2}}} \,,
\end{align}

where we have used that $\Omega_{\beta,n}=\Omega_{\beta,-n}$ and introduced a factor $1/2$ to compensate for the doubling in the sum. We want to use this expression to extract the asymptotic expansion as $\lambda\rightarrow 0$ and then study the Borel-Laplace summation of the obtained series. This form of the expression makes it clear that it is invariant under $\lambda\rightarrow -\lambda$ and hence the asymptotic expansion should only contain even powers of $\lambda$. We therefore want to study  the asymptotic expansion of 

$$\frac{ \exp(\ii  n k\lambda+ \exp(-\ii  n k\lambda)}{{\left( 2 \sin \left( \frac{k\lambda}{2}\right)\right)^{2}}}\,.$$

We recall the generating function of Bernoulli polynomials:
\begin{equation}
\frac{w e^{w x}}{ e^w-1} = \sum_{m=0}^{\infty} \frac{B_m(x)}{m!} w^m \,, 
\end{equation}
by taking a logarithmic derivative $w\partial_w$ of both sides we obtain:
\begin{equation}
\frac{w e^{wx}}{ e^w-1}+ \frac{w^2 x e^{w x}}{ e^w-1}-\frac{w^2  e^{w(x+1)}}{ (e^w-1)^2} = \sum_{m=1}^{\infty} \frac{B_m(x)}{(m-1)!} w^m \,, 
\end{equation}
by re-arranging we get:
\begin{align}
\frac{e^{w(x+1)}}{ (e^w-1)^2} &= \frac{1}{w^2}(1+w x)- \sum_{m=2}^{\infty} \frac{B_m(x)}{m (m-2)!} w^{m-2}+  x \sum_{m=1}^{\infty} \frac{B_m(x)}{m!} w^{m-1} \,,\nonumber\\
&= \frac{1}{w^2}(1+w x)- \sum_{m=2}^{\infty} \frac{1}{m!} \left( (m-1) B_m(x)- x m B_{m-1}(x)\right)  w^{m-2}\,, 
\end{align}

we thus obtain for 

\begin{align}\label{eq:symBern}
&\frac{e^{w(x+1)}+e^{w(-x+1)}}{ (e^w-1)^2} = \frac{2}{w^2}+ \left( -\frac{1}{6}+x^2\right)  \nonumber \\ 
&- \sum_{g=2}^{\infty}  \frac{1}{2g!} \left(  (2g-1) (B_{2g}(x)+B_{2g}(-x))- 2g x (B_{2g-1}(x)-B_{2g-1}(-x))\right)  w^{2g-2}\,.
\end{align}

We use this for the expression
$$\frac{ \exp(\ii  n k\lambda+ \exp(-\ii  n k\lambda)}{{\left( 2 \sin \left( \frac{k\lambda}{2}\right)\right)^{2}}} =- \frac{ \exp(\ii  (n+1) k\lambda)+\exp(\ii  (-n+1) k\lambda)}{{\left(  e^{i k\lambda}-1\right)^{2}}} $$

 with $w=\ii k \lambda$ and $x=n$, to obtain:

\begin{align}  \label{eq:asymptBern}
&\frac{ \exp(\ii  n k\lambda+ \exp(-\ii  n k\lambda)}{{\left( 2 \sin \left( \frac{k\lambda}{2}\right)\right)^{2}}}=  \frac{2}{k^2 \lambda^2}+ \left( \frac{1}{6}-n^2\right) \nonumber\\
 & +\sum_{g=2}^{\infty}  \frac{1}{2g!} \left(  (2g-1) (B_{2g}(n)+B_{2g}(-n))-2g n (B_{2g-1}(n)-B_{2g-1}(-n))\right)  (\ii k \lambda)^{2g-2}
 \end{align}
 
we thus get

\begin{align}\label{eq:FGVasmpt}
F_{GV}(\lambda,t) &\sim 
  \sum_{\beta>0} \sum_{n \in \mathbb{Z}} \frac{\Omega_{\beta,n}}{2}   \sum_{k\ge 1}   \frac{Q^{k\beta}}{k}  \times \Bigg( \frac{2}{k^2 \lambda^2}+ \left( \frac{1}{6}-n^2\right) \nonumber\\
 & +\sum_{g=2}^{\infty}  \frac{1}{2g!} \left( (2g-1) (B_{2g}(n)+B_{2g}(-n))- 2g n (B_{2g-1}(n)-B_{2g-1}(-n))\right)  (\ii k \lambda)^{2g-2} \Bigg) \,   \,.
\end{align}

We noted in Sec.~\ref{subsec:sheafinv} that the sum over $n$ of all $\Omega_{\beta,n}$ is given by the limit of $y\rightarrow 1$ of
\begin{equation*}
\sum_{n\in \mathbb{Z}\,} \Omega_{\beta,n} y^n = \sum_{g\ge0} [\textrm{GV}]_{\beta,g}\, (-1)^g (y^{1/2}-y^{-1/2})^{2g}\,,
\end{equation*}
giving
\begin{equation}\label{eq:sumOmega}
\sum_{n\in \mathbb{Z}} \Omega_{\beta,n}  = [\textrm{GV}]_{\beta,0}\,,
\end{equation}
similarly by taking two logarithmic derivatives with respect to $y$ and then taking the limit we obtain:
\begin{equation}\label{eq:sumOmega}
\sum_{n\in \mathbb{Z}} n^2 \Omega_{\beta,n}  = 2 [\textrm{GV}]_{\beta,1}\,.
\end{equation}

Both of these equalities can be used to express 
\begin{align}\label{eq:FGVasmpt}
F_{GV}(\lambda,t) &\sim \frac{1}{\lambda^2}\tilde{F}^0(t) + \tilde{F}^{1}(t) + \sum_{g=2}^{\infty}  \lambda^{2g-2} \tilde{F}^g(t) \,,
\end{align}
where

$$ \tilde{F}^0(t) =  \sum_{\beta>0} \sum_{n \in \mathbb{Z}} \Omega_{\beta,n}   \sum_{k\ge 1}   \frac{Q^{k\beta}}{k^3} =  \sum_{\beta>0} [\textrm{GV}]_{\beta,0} \operatorname{Li}_{3}(Q^{\beta})  $$

and 
the constant term in $\lambda$ in the expansion \eqref{eq:FGVasmpt} is:
\begin{align}
\sum_{\beta>0} \sum_{n\in\mathbb{Z}}  \Omega_{\beta,n}     \frac{1}{2}   \left( n\, 2\,  B_{1}(n)-  B_{2}(n)\right)    \operatorname{Li}_{1}(Q^{\beta})  &= \sum_{\beta>0} \sum_{n\in\mathbb{Z}}  \Omega_{\beta,n}    \frac{1}{12} (-1+ 6 n^2)    \operatorname{Li}_{1}(Q^{\beta})\nonumber  \\
&= \sum_{\beta>0}   \left( [\textrm{GV}]_{\beta,1} -\frac{1}{12}  [\textrm{GV}]_{\beta,0} \right)    \operatorname{Li}_{1}(Q^{\beta}) \nonumber\\
&= \tilde{F}^1(t)\,,
\end{align}
and 

\begin{equation}
\tilde{F}^g(t)= \sum_{\beta>0} \sum_{n\in\mathbb{Z}}\frac{(-1)^{(g-1)} \Omega_{\beta,n}}{2\cdot2g!} \left( (2g-1) (B_{2g}(n)+B_{2g}(-n))- 2g n (B_{2g-1}(n)-B_{2g-1}(-n))\right)  \operatorname{Li}_{3-2g}(Q^{\beta}) .
\end{equation}
\end{proof}

\subsection{The Borel transform}\label{boreltransproof}

\begin{thm} \label{thm:Borel}
Let $\Phi(\check{\lambda},t):= \sum_{g=2}^{\infty} \lambda^{2g-2} \tilde{F}^g(t)=  \sum_{g=2}^{\infty} \check{\lambda}^{2g-2} \, (2\pi)^{2g-2} \tilde{F}^g(t)\,, \check{\lambda}=\lambda/2\pi$ denote the asymptotic series of the generating function of  higher genus $g \ge 2$ GW invariants. Let furthermore the formal variables $t^{\beta} \in \mathbb{C}^{\times}\setminus \mathbb{Z}$ be with $\textrm{Re}(t^\beta) <1/2$ for all $\beta \in H_2(X,\mathbb{Z})$, then the Borel transform $G(\xi,t)$ of $\Phi(\check{\lambda},t)$  can be analytically continued to a meromorphic function given by:
\begin{equation} \label{eq:BorelTransform}
G(\xi,t) 
=  \frac{1}{\xi} \sum_{\beta>0} \sum_{k,n\in \mathbb{Z}} \frac{\Omega_{\beta, n}}{2}
\left(
  \frac{2(t^{\beta}+k)^2}{\xi^2}
  + \left(n^2 - \frac{1}{6}\right)
  - \frac{
      e^{\frac{\xi(n+1)}{t^{\beta}+k}}
      + e^{\frac{\xi(-n+1)}{t^{\beta}+k}}
    }{
       \left( e^{\frac{\xi}{t^{\beta}+k}} - 1 \right)^2
    } 
\right)\,,
\end{equation}
with poles at $\xi= 2\pi \ii m  (t^\beta+k)\,,$ for $m\in \mathbb{Z}\setminus \{0\}$ and $k\in \mathbb{Z}$.
\end{thm}

\begin{proof}

The Borel transform is defined as the formal power series $G(\xi,t):=\mathcal{B}(\Phi(-,t))(\xi)$, where
\begin{equation}
    \mathcal{B}\colon\check{\lambda}\mathbb{C}[[\check{\lambda}]]\to \mathbb{C}[[\xi]], \;\;\;\; \mathcal{B}(\check{\lambda}^{n+1})=\frac{\xi^n}{n!}.
\end{equation}
We wish to study the Borel transform of 
\begin{align}
\Phi(\check{\lambda},t)& = \sum_{g=2}^{\infty}\sum_{\beta>0} \sum_{n\in\mathbb{Z}}\frac{(2\pi \ii)^{2g-2} \Omega_{\beta,n}}{2\cdot2g!} \left( (2g-1) (B_{2g}(n)+B_{2g}(-n))- 2g n (B_{2g-1}(n)-B_{2g-1}(-n))\right)\nonumber \\
&\times  \operatorname{Li}_{3-2g}(Q^{\beta}) \check{\lambda}^{2g-2}\,,
\end{align}
which is given by
\begin{align}
G(\xi,t)& = \sum_{g=2}^{\infty}\sum_{\beta>0} \sum_{n\in\mathbb{Z}}\frac{(2\pi \ii)^{2g-2} \Omega_{\beta,n}}{2\cdot2g!} \left( (2g-1) (B_{2g}(n)+B_{2g}(-n))- 2g n (B_{2g-1}(n)-B_{2g-1}(-n))\right)\nonumber \\
&\times  \operatorname{Li}_{3-2g}(Q^{\beta}) \frac{\xi^{2g-3}}{2g-3!}\,,
\end{align}

we now use the identity
$$
\operatorname{Li}_s \left(e^{\mu}\right)
=
\Gamma(1 - s)
\sum_{k=-\infty}^{\infty}
\left(2k\pi i - \mu\right)^{\,s - 1}
$$

$$
\text{valid for } \textrm{Re}(s) < 0 \quad \text{and all } \mu \text{ such that } e^{\mu} \neq 1,
$$
to write

\begin{align}
G(\xi,t)& = \frac{1}{\xi} \sum_{g=2}^{\infty}\sum_{\beta>0} \sum_{k\in \mathbb{Z}}\sum_{n\in\mathbb{Z}} \frac{\Omega_{\beta,n}}{2\cdot (2g)!} \left( (2g-1) (B_{2g}(n)+B_{2g}(-n))- 2g n (B_{2g-1}(n)-B_{2g-1}(-n))\right)\nonumber \\
& \times \left(\frac{\xi}{t^{\beta}+k}\right)^{2g-2}\,,
\end{align}
we can now use \eqref{eq:symBern} to write this as:
\begin{equation}
G(\xi,t)= \frac{1}{\xi} \sum_{\beta>0} \sum_{k,n\in \mathbb{Z}} \frac{\Omega_{\beta, n}}{2}
\left(
  \frac{2(t^{\beta}+k)^2}{\xi^2}
  + \left(n^2 - \frac{1}{6}\right)
  - \frac{
      e^{\frac{\xi(n+1)}{t^{\beta}+k}}
      + e^{\frac{\xi(-n+1)}{t^{\beta}+k}}
    }{
       \left( e^{\frac{\xi}{t^{\beta}+k}} - 1 \right)^2
    } 
\right)\,.
\end{equation}
\end{proof}

\subsection{Borel sums and Stokes factors}

Let  $\Gamma= H_0(X,\mathbb{Z})\oplus H_2(X,\mathbb{Z})$ and let the constant $1$ be the generator of $H_0(X)$ and $\beta\in H_{2}(X,\mathbb{Z})$, we denote the central charge of $\gamma\in\Gamma$ by $Z(\gamma)$ , let $\gamma_{\beta,k}=(k,\beta)$ and its  the central charge $Z(\gamma_{\beta,k})=2\pi\ii (t^{\beta}+k)$ and assume for the moment that there are as many formal variables $t^{\beta}$ as there are curve classes. Let   $t^{\beta}\in \mathbb{C}-\mathbb{Z}$ and  $k\in \mathbb{Z}$ let $l_{\gamma}:=\mathbb{R}_{<0}\cdot Z(\gamma)$ and $l_{\infty}:=\I\mathbb{R}_{<0}$. 
\begin{dfn}
Given any ray $\rho$ from $0$ to $\infty$ different from $\{\pm l_{\gamma_{\beta,k}}\} \cup \{\pm l_{\infty}\}$, and $\lambda$ in the half-plane $\mathbb{H}_{\rho}$ centered at $\rho$. We define the Borel sum of $\widetilde{F}(\lambda,t)$ along $\rho$ as
\begin{equation}\label{BorelS}
    F_{\rho}(\lambda,t):=\ \frac{1}{\lambda^2} \widetilde{F}^0(t) + \widetilde{F}^1(t) +\int_{\rho}\mathrm d\xi\, e^{-\xi/\check{\lambda}} G(\xi,t)\,.
\end{equation}
\end{dfn}

\begin{thm} \label{thm:Stokes}

Let $\rho_{\beta,k}$ be a ray in the sector determined by the rays of singularities $l_{\gamma_{(\beta,k)}}$ and $l_{\gamma_{(\beta,k-1)}}$ and $\rho_{\beta,k} \ne l_{\gamma_{(\beta',m)}}$ for all $\beta' \ne \beta$ and all $m \in \mathbb{Z}$ . Then if $\mathrm{Im}(t^{\beta})>0$ and furthermore $t^{\beta'} \ne t^{\beta}$ for all $\beta'\ne \beta$, on the overlap of their domains of definition in the $\lambda$ variable we have
\begin{align}\label{eq:Stokesjumpsgeneral}
    \phi_{\pm l_{\beta,k}}(\lambda,t)&:=F_{\pm \rho_{\beta,k+1}}(\lambda,t)
    -F_{\pm \rho_{\beta,k}}(\lambda,t) =\frac{[GV]_{\beta,0}}{2\pi \I}\partial_{\check\lambda}\Big(\check{\lambda}\,\mathrm{Li}_2\big(e^{\pm 2\pi \I(t^{\beta}+k)/\check \lambda}\big)\Big)\;, \, \\
    &= \frac{[GV]_{\beta,0}}{2\pi \I} \left( \mathrm{Li}_2\big(e^{\pm 2\pi \I(t^{\beta}+k)/\check \lambda}\big) \pm 2\pi \I \frac{(t+k)}{\check \lambda} \log \left(1- e^{\pm 2\pi \I(t^{\beta}+k)/\check \lambda} \right)\right)\;.
\end{align}
Defining
\begin{equation}
 F_{Stokes}(\lambda,t): =  \sum_{\beta>0} \sum_{k=0}^{\infty} \phi_{ l_{\beta,k}}(\lambda,t)
 \end{equation}
 we have:
 
 \begin{equation}
 F_{Stokes}(\lambda,t) = - F_{top,sg}(\lambda,t)=  \frac{1}{2\pi} \partial_{\lambda} \left( \lambda \, \tilde{F}^{0}_{\textrm{def}} \left( \frac{2\pi}{\check{\lambda}}, \frac{t-1/2}{\check{\lambda}} \right)\right)\,,
 \end{equation}
 in terms of the deformation of the prepotential studied in sec.~\ref{sec:defprepot}.
 
\end{thm}

\begin{proof}
Let $y_{m}=\xi-2\pi \mathrm i m(t^{\beta}+k)$ be a local coordinate near the pole of the Borel transform $G(\xi,t)$ given by $\xi=2\pi \mathrm i m (t^{\beta}+k)$, then by Taylor expanding the integrand of the Borel sum near $y_m=0$ we obtain
\begin{align}
   e^{-\xi/\check{\lambda}} G(\xi,t)&= \sum_{n\in \mathbb{Z}} \Omega_{\beta,n} e^{-2\pi \mathrm i m(t^{\beta}+k)/\check{\lambda}} \left( \frac{1}{y_m^2} \frac{\mathrm i (t^{\beta}+k)}{2\pi m}+\frac{1}{y_m}\left( - \frac{1}{4\pi^2 m^2 }- \frac{\mathrm  i}{2\pi m\check{\lambda}} (t^{\beta}+k)\right) + \mathcal{O}(1)\right)  \nonumber \\ &+   \sum_{\beta'\ne \beta}  \Omega_{\beta',n} \times \mathcal{O}(1) \,.
\end{align}

The expressions for Stokes jumps 
$$  \phi_{\pm l_{\beta,k}}(\lambda,t)=F_{\pm \rho_{\beta,k+1}}(\lambda,t)
    -F_{\pm \rho_{\beta,k}}(\lambda,t) $$
    
   define Hankel contours which can be evaluated using the residue theorem to obtain:
    
     \begin{equation}
  \phi_{\pm l_{\beta,k}}(\lambda,t) =  \sum_{n\in \mathbb{Z}} \frac{\Omega_{\beta,n}}{2\pi \I} \left( \mathrm{Li}_2\big(e^{\pm 2\pi \I(t^{\beta}+k)/\check \lambda}\big) \pm 2\pi \I \frac{(t+k)}{\check \lambda} \log \left(1- e^{\pm 2\pi \I(t^{\beta}+k)/\check \lambda} \right)\right)\;,
\end{equation}    

the sum $\sum_{n\in \mathbb{Z}} \Omega_{\beta,n} =[GV]_{\beta,0} $ yields again 

\begin{equation}
  \phi_{\pm l_{\beta,k}}(\lambda,t) = \frac{[GV]_{\beta,0}}{2\pi \I} \left( \mathrm{Li}_2\big(e^{\pm 2\pi \I(t^{\beta}+k)/\check \lambda}\big) \pm 2\pi \I \frac{(t+k)}{\check \lambda} \log \left(1- e^{\pm 2\pi \I(t^{\beta}+k)/\check \lambda} \right)\right)\;.
\end{equation}

We now have for 

\begin{align}
 F_{Stokes}(\lambda,t) &=  \sum_{\beta>0} \sum_{k=0}^{\infty} \phi_{ l_{\beta,k}}(\lambda,t) = \sum_{\beta>0}  \sum_{k=0}^{\infty} \frac{[\textrm{GV}]_{\beta,0}}{2\pi i} \partial_{\check{\lambda}} \left( \check{\lambda} \, \operatorname{Li}_2(e^{2\pi \ii (t^{\beta}+k)/\check{\lambda}})\right) \,, \nonumber\\
  &=\frac{1}{2\pi} \partial_{\lambda} \left( \lambda \, \tilde{F}^{0}_{\textrm{def}} \left( \frac{2\pi}{\check{\lambda}}, \frac{t-1/2}{\check{\lambda}} \right)\right) \,,
\end{align}
using the proof of proposition \ref{prop:correc}.

\end{proof}

\begin{rem}
\begin{itemize}
\item Equation \eqref{eq:Stokesjumpsgeneral} is the generalisation of the expressions obtained in \cite{Alim:2021mhp}. Via the correspondence of the resurgence and the Riemann-Hilbert problem put forward by Bridgeland in \cite{Bridgeland1,BridgelandCon},  discussed in \cite[Section 4]{Alim:2021mhp}the prefactors $[GV]_{\beta,0}$ correspond to the DT invariants of an underlying DT RH problem. The expected RH problem associated to the part of the GW generating function which corresponds to the genus $0$ part in the GV expansion has been addressed in \cite[Section 6]{Bridgeland1}.

\item The form of the Stokes jumps \eqref{eq:Stokesjumpsgeneral} put forward in \cite{Alim:2021mhp} was already found in a slightly different guise in earlier work of Pasquetti and Schiappa\cite{Pasquetti:2009jg}, it was furthermore identified in \cite[Prop. 4.7]{Alim:2022oll} as the quantum dilogarithm function with equal arguments and more recently this form was also expected to hold more generally in topological string theory in \cite{Iwaki:2023cek}.

\end{itemize}
\end{rem}

\subsection{Singularities in the Borel plane}

The singularities of the Borel transform lie along the rays  $l_{\gamma}\,,\gamma \in \Gamma$, whose slope is obtained from the central charges $Z(\gamma)$. In the general discussion above we have assumed that there are as many formal variables $t^{\beta}$ as curve classes which allowed the straightforward generalisation of \cite{Alim:2021mhp}. From mirror symmetry however we know that all $t^{\beta}$ should live in an $n=h^{1,1}(X)$ complex manifold, which is the moduli space of complexified K\"ahler forms of $X$.
Let therefore $C_a \in H_{2}(X,\mathbb{Z}), a=1,\dots,n$ be a set of curve classes spanning $H_2(X,\mathbb{Z})$, then in general all curve classes $\beta$ can be written as $\beta=\sum_{a=1}^{n} d_a C_a$, where the $d$'s are the degrees of the curve. The formal variables $t^{\beta}$ become $t^{\beta}=\sum_{a=1}^{n} d_a t^a$ and the $t^a$ correspond to local coordinates on the moduli space of complexified K\"ahler forms of the underlying CY $X$. We can now see that this leads to regions in the Borel plane where there are accumulation rays of singularities. 

We first note that the singularities of the Borel transform \eqref{eq:BorelTransform} are at 

\begin{equation}\label{Borelsing}
\xi= 2\pi \ii m (\sum_{a=1}^{n} d_a t^a + k)\,, k\in\mathbb{Z}\,, \quad d_a \in \mathbb{N}\,, \quad a=1,\dots,n\,,
\end{equation}
as long as $[GV]_{\{d_a\}  ,0}\ne0$.

For the purpose of illustrating the loci of the singularities we focus on a single curve class $C$ and its higher degrees given by $d\in\mathbb{N}$, the singularities of the Borel transform coming from that curve are on the rays:
\begin{equation}
\pm l_{d,k}:=\pm\left( -\mathbb{R}_{<0} 2\pi \ii (dt + k) \right)\,, \quad d\in\mathbb{N},\quad  k\in \mathbb{Z}\,,
\end{equation}
as long as $[GV]_{d,0}\ne 0$.

The slope of the line $ l_{d,k}$ is given by:
\begin{equation}
- \frac{\textrm{Re}(d t)+ k}{\textrm{Im}dt}=  - \frac{\textrm{Re}( t)}{\textrm{Im} t} - \frac{k/d}{\textrm{Im}t}\,,
\end{equation} 
we observe thus, that in addition to the accumulation rays of singularities $\pm l_{\infty}$ due to the growth in $k$, there are further accumulation rays due to the increase in the degree around $\pm l_{1,0}$. Furthermore we can expect generically rays of singularities along all $\pm l_{1,0} + \mathbb{Q}\cdot \frac{1}{\textrm{Im}t}$. An illustration of the some of the rays of singularities is give in Fig.~\ref{fig:Borelsing}.

A distinguished set of singularities in the Borel plane thus corresponds to the loci where only one $d_a=1$ and $m=\pm 1$ in \eqref{Borelsing}. The corresponding Stokes factor would then have the integer prefactor given by $[GV]_{d_a=1,0}$. In the case of the quintic there is only one curve class and the degree $d=1$ curve GV count is given by $2875$. This number was recently seen numerically using the interplay of large order genus expansion and resurgence in \cite{Gu:2023mgf}.

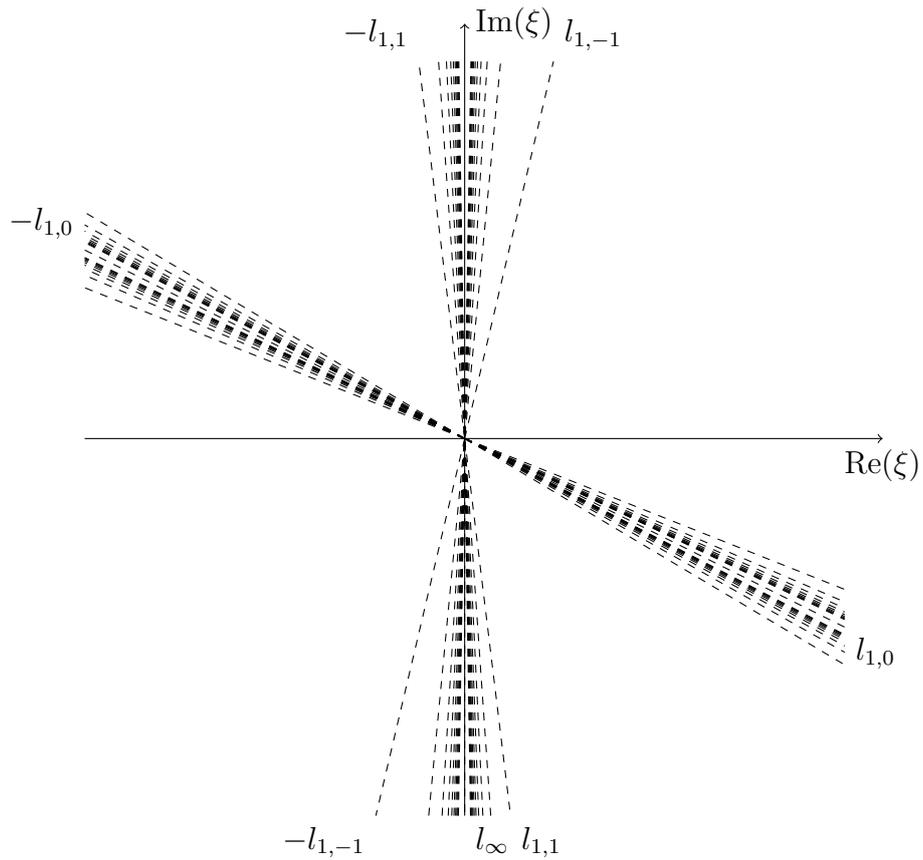
\begin{figure}
\begin{center}
\begin{tikzpicture}
 \draw[->] (-5,0) -- (5.5,0) coordinate[label = {below:$\textrm{Re} (\xi)$}] (xmax);
  \draw[->] (0,-5) -- (0,5.5) coordinate[label = {right:$\textrm{Im} (\xi)$}] (ymax) ;
 \draw[dashed] (0,0)--(5,-2.5) node[anchor=north west] {$l_{1,0}$};
\draw[dashed] (0,0)--(-5,2.5) node[anchor=south east] {$-l_{1,0}$};
  \draw[dashed] (0,0)--({5/(-2-6.28)},5) node[anchor=south east] {$-l_{1,1}$};
  \draw[dashed] (0,0)--({-5/(-2-6.28)},-5) node[anchor=north west] {$l_{1,1}$}; 
 \draw[dashed] (0,0)--({5/(-2-6.28*(-1))},5) node[anchor=south west] {$l_{1,-1}$};
  \draw[dashed] (0,0)--({-5/(-2-6.28*(-1)},-5) node[anchor=north east] {$-l_{1,-1}$};
 \draw[dashed] (0,0)--(0,-5) node[anchor=north west] {$l_{\infty}$}; 
 \foreach \k in {-10,...,-2} {
  \draw[dashed] ({-5/(-2-6.28*\k)},-5)--({5/(-2-6.28*\k)},5);
  };
 \foreach \k in {2,...,10} {
  \draw[dashed] ({-5/(-2-6.28*\k)},-5)--({5/(-2-6.28*\k)},5);
  };
  \foreach \d in {2,...,10} {
  \draw[dashed] (0,0)--(5,{-2.5+1/(\d)}) ;
   \draw[dashed] (0,0)--(5,{-2.5-1/(\d)}) ;
    \draw[dashed] (0,0)--(-5,{2.5+1/(\d)}) ;
   \draw[dashed] (0,0)--(-5,{2.5-1/(\d)}) ;
   };

\end{tikzpicture}
\end{center}

\caption{Illustration of some of the singularity rays in the Borel plane $l_{d,k}=\mathbb{R}_{<0}\cdot 2\pi \I (d t+k)$ in the Borel plane, plotted for $t= \frac{1}{\pi}\left(\frac{1}{2}+ \I\right)$ and $(d,k)=(1,-10),\dots,(1,10)$ as well as $(d,\pm1)= (2,\pm1)\dots (10,\pm1)$}
\label{fig:Borelsing}
\end{figure}


\section{Discussion and Conclusions}\label{sec:conclusions}
In this work we have put forward an expression for the non-perturbative topological string partition function on any Calabi-Yau threefold, compact or non-compact, in the holomorphic limit which provides the asymptotic expansion of topological strings encoding the enumerative GW as well as GV invariants of the underlying CY threefold. A crucial step in obtaining this expression was identifying the partition function associated to any CY manifold as a product over basic building blocks which correspond to the constant map contributions on the one hand while on the other hand the non-constant maps are captured by resolved conifold partition functions with shifted arguments, raised to the power of sheaf invariants. While the general convergence properties of the full non-perturbative expression requires more knowledge of the sheaf invariants which are obtained from all higher genus GV invariants, it was found that the non-perturbative corrections are governed by only the genus 0 GV invariants, which via mirror symmetry are encoded in rational functions on the complex structure moduli space of the mirror manifold and have thus good convergence properties. It was furthermore shown that the non-perturbative corrections could be entirely captured in terms of a single function which was identified as an $\varepsilon$ deformation of the genus 0 prepotential. This function corresponds in the case of the resolved conifold to the refined topological string in the NS limit which was shown to capture the non-perturbative corrections in that case \cite{Hatsuda:2015owa,Alim:2021ukq}.

At the end it is perhaps possible to speculate along some future directions and questions:
\begin{itemize}

\item \emph{Global properties}
The non-perturbative expression for the topological string partition function in this work is obtained from the holomorphic limit which contains the enumerative geometry information of the underlying CY threefold. On the A-model side, this is the large radius region, on the mirror B-side this is the vicinity of a point of maximal unipotent monodromy. It is known, for instance from the computations using the topological vertex \cite{Aganagic:2003db}, that the obtained expression does not give a good global function on the underlying moduli space and it is in general difficult, if not impossible to analytically continue. The route to obtain global expressions for the topological string partition functions is using the holomorphic anomaly equations \cite{Bershadsky:1993cx} together with the recursive polynomial construction of the genus expansion put forward in \cite{Yamaguchi:2004bt,Alim:2007qj} as was recently revisited in connection to non-perturbative topological strings in \cite{Gu:2023mgf}. In \cite{Iwaki:2023cek}, the approach of making a transseries ansatz and using the holomorphic anomaly as studied in \cite{Couso-Santamaria:2013kmu,Couso-Santamaria:2014iia} as well as the results of \cite{IwakiPainleve,Gu:2022sqc,Gu:2023mgf} were used to put forward an expectation for the general form of the Stokes automorphisms acting on the topological string partition function. The expected Stokes factors are of the form obtained from the analytic results of \cite{Pasquetti:2009jg,Alim:2021mhp} as well as their generalizations in this work, obtained without use of the holomorphic anomaly. One may expect that the non-holomorphic topological string has an even richer non-perturbative structure. In fact, in \cite{Alim:2015qma} the holomorphic anomaly equations were used together with the polynomial structure of higher genus topological strings to derive an Airy differential equation in the topological string coupling governing a universal piece of the higher genus free energies for all CY families, in addition to the perturbative piece captured by one solution of the Airy differential equation, the second solution suggests corrections with $e^{-1/\lambda^2}$ behaviour, this was explained by the effects of NS-5 branes in \cite{Couso-Santamaria:2015hva}. These non-perturbative effects seem to be completely missing in the results of the current paper which is likely an artefact of starting with the asymptotic series of topological strings in the limit leading to the enumerative geometry content.

\item \emph{Non-perturbative absence of the gap?}
Although the non-perturbative expression of this work may not be the full non-holomorphic one, one may however speculate at this stage, whether the \emph{gap condition} used in \cite{Huang:2006hq} and which refers to a leading singularity in the topological string free energies $\mathcal{F}^g(t)$ near a point in moduli space where finitely many physical states become massless is an artefact of the perturbative expansion. Near the conifold point of the resolved conifold, i.~e.~ $t\rightarrow 0$, it was proven in \cite{Alim:2021ukq} that the analytic continuation in $t$ order by order in $\lambda$ gives a leading singular behavior of the expected form and that moreover the constant term of this analytic continuation contained the information of the constant map contribution. A surprise of the non-perturbatively complete Borel sums put forward in \cite{Alim:2021mhp} is that their limit $t\rightarrow 0$ is well behaved and it moreover encodes the constant map contributions. Furthermore, the combination of the non-constant and constant map contributions leads to a partition function which is canonically normalized as $t\rightarrow 0$. One may speculate if a similar phenomenon also appears for topological strings in general. Hints that the moduli space to be considered is a non-trivial combination of $\lambda$ and $t$ already appear in the analysis of the natural dual, strongly coupled variables $\frac{t}{\lambda},\frac{1}{\lambda}$. The expressions put forward in this work in terms of the triple sine function have again good $t^{\beta}\rightarrow 0 $ limits. In general however the conifold loci correspond to the vanishing of a linear combination of dual period which may not be captured by the current expressions. Another speculation is whether the topological string partition functions are normalized in a canonical way as well or whether this was an accident of the resolved conifold, the general case would require non-trivial relations between $\chi(X)$ the Euler character of the CY and sums of its GV invariants.

\item \emph{Quantum Picard-Fuchs equations?} The difference equation of the asymptotic expansion of the resolved conifold free energy of \cite{Alim:2020tpw} was interpreted in \cite{Alim:2022oll} as giving a closed expression of the quantum Picard-Fuchs equations which refer to the differential equations annihilating the Voros symbols in the exact WKB analysis of the quantum curves obtained from the mirrors of non-compact CY.  It was recently shown in \cite{Jockers:2024zio} that the quantum Picard-Fuchs equations considered by \cite{Jockers:2019wjh} in the context of quantum K-theory give, for the resolved conifold, exactly the difference equations obtained in \cite{Alim:2020tpw}. It would be interesting to develop such a correspondence more generally, perhaps guided by the expressions of the deformed prepotential put forward in this work.

\item \emph{The Nekrasov-Shatshvili limit?} The deformed prepotential introduced in this work, which is capturing the non-perturbative corrections to the topological string coincides with the Nekrasov-Shatashvili limit of the refined topological string in the case of the resolved conifold. It is not entirely clear whether such a relation can be expected more generally. What is clear however, is that the piece coming from the deformed prepotential in the non-perturbative expression is responsible for the cancellation of poles of the topological string free energy in $\lambda$ and as such is playing the same role as the NS limit in the Topological String/Spectral Theory correspondence of \cite{Grassi:2014zfa}. Further connections between the deformed prepotential and the NS limit can probably understood from a careful study of the definition of the NS generating function in terms of refined DT invariants as in \cite{Bousseau:Anomaly}. A limit of the function introduced there coincides with the deformed prepotential.

\item \emph{Relations between enumerative invariants} the expressions for the GW generating functions in terms of the sheaf invariants given in prop. \ref{prop:GWinsheaf} suggest new relations between GW and sheaf invariants. Perhaps some of these can be understood in the context of the relations between sheaf and GW invariants as studied for instance in \cite{Bousseau:Takahashi} \footnote{I would like to thank the referee for pointing out a possible relation.}

\end{itemize}


\subsection*{Acknowledgements}
I would like to thank Lotte Hollands, Arpan Saha, Ivan Tulli and Joerg Teschner for collaboration on projects which led to this work as well as for discussions and comments. I would like to thank Noah Dizep, Veronica Fantini, Stavros Garoufalidis, Albrecht Klemm, Maxim Kontsevich, Marcos Marino,  Andy Neitzke, Yan Soibelman and Campbell Wheeler for discussions and comments. I would like to furthermore thank the organizers of the Les Diablerets workshop on Quantum Moduli in September 2023 and the organizers of the INI workshop on Number Theory, Machine Learning and Quantum Black Holes in October 2023 for the opportunity to announce the results of this work. I would also like to thank the referees for valuable insights and suggestions. This work is supported through the DFG Emmy Noether grant AL 1407/2-1.

\begin{appendix}\label{appendix}

\section{Multiple sine functions}\label{specialfunctapp}
We introduce several special functions which are used throughout the paper following the notation of \cite{Narukawa}.  Suppose $\omega_1\,,\dots,\omega_r \in \mathbb{C}$ all lie on the same side of some straight line through the origin, then the multiple zeta function is defined by the series
\begin{equation}
\zeta_r(s,z|\underline{\omega}) = \sum_{n_1,\dots,n_r=0}^{\infty} \frac{1}{(n_1\omega_1+\dots+ n_r \omega_r +z)^s}\,,
\end{equation}
for $z\in \mathbb{C}$ and $\textrm{Re}(s) >r$, where the exponential is rendered one-valued. This series is holomorphic in the domain $\{ \textrm{Re} (s)>r\}$ and is analytically continued to $s\in \mathbb{C}$. 

Since it is holomorphic at $s=0$ we can define the multiple gamma function by
\begin{equation}
\Gamma_r(z | \underline{\omega}) =\exp \left( \frac{\partial}{\partial s} \zeta_r(s,z | \underline{\omega})|_{s=0} \right),
\end{equation}
where
$ |\underline{\omega}| = \omega_1 + \dots + \omega_r\,, $
and the multiple sine function is given by
\begin{equation}
\sin_r(z|\underline{\omega}) = \Gamma_r(z|\underline{\omega})^{-1}\, \Gamma_r(|\underline{\omega}| -  z | \underline{\omega})^{(-1)^r}\,.
\end{equation}
This satisfies the functional equations:
\begin{eqnarray}\label{eq:sinefunc}
&&\sin_r(z+\omega_j | \underline{\omega}) =\sin_{r-1} (z | \underline{\omega}^{-}(j))^{-1} \sin_r (z|\underline{\omega}) \, , \\
&&\sin_r(z| \underline{\omega}) \,\sin_r(|\underline{\omega}|-z | \underline{\omega})^{(-1)^r}=1\,,
\end{eqnarray}
where $$ \underline{\omega}^{-}(j)= (\omega_1, \dots,\overset{\vee}{\omega_j},\dots,\omega_r)\,,$$
and $\overset{\vee}{\omega_j}$ means omitting $\omega_j$.

For $z\in \mathbb{C}$, $\underline{\omega}=(\omega_1,\dots,\omega_r)$, $\omega_j \in \mathbb{C}\setminus\{0\}$ we define the multiple Bernoulli polynomials $B_{r,n}(z,\underline{\omega})$ in terms of the generating function
\begin{equation}\label{eq:Bernoullipolgen}
\mathcal{B}_r(z,x|\, \omega_1,\,\dots,\omega_r) = \frac{x^r e^{z x}}{\prod_{j=1}^r (e^{\omega_j x}-1)} = \sum_{n=0}^{\infty} B_{r,n}(z|\underline{\omega}) \frac{x^n}{n!}\,.
\end{equation}

The multiple sine functions are defined using the Barnes multiple Gamma functions \cite{Barnes}. For a variable $z\in \mathbb{C}$ and parameters $\omega_1,\ldots,\omega_r \in \mathbb{C}^{*}$ these are defined by
\begin{equation}
    \sin_r(z\,|\, \omega_1,\dots,\omega_r):= \Gamma_{r}(z\, |\, \omega_1,\dots,\omega_r)^{-1} \cdot \Gamma_{r}\left(\sum_{i=1}^r \omega_i - z\, |\, \omega_1,\dots,\omega_r\right)^{(-1)^r} \,.
\end{equation}
For further definitions see e.g.~\cite{BridgelandCon,Ruijsenaars1} and references therein. 

Another special function is (related to) Fadeev's quantum dilogarithm
\begin{equation}\label{g2def}
    \mathcal{S}_2(z\, | \, \omega_1,\omega_2) := \exp\left(-\frac{\pi \I}{2} \cdot B_{2,2}(z|\,\omega_1,\omega_2)\right) \cdot \sin_2(z\, |\, \omega_1,\omega_2),
\end{equation}
which we have introduced in terms of the double sine function as in \cite{BridgelandCon}. 

The final special function is defined by
\begin{equation}\label{g3def}
    \mathcal{S}_3(z\, | \, \omega_1,\omega_2) := \exp\left(\frac{\pi \I}{6} \cdot B_{3,3}(z+\omega_1\,|\,\omega_1,\omega_1,\omega_2)\right) \cdot \sin_3(z+\omega_1\, |\, \omega_1,\omega_2,\omega_3),
\end{equation}

\end{appendix}


\providecommand{\href}[2]{#2}\begingroup\raggedright\endgroup

\end{document}